\documentclass[letter]{article}

\usepackage{amsthm}
\usepackage{times}
\usepackage{url}
\usepackage{epsfig}
\usepackage{amsmath}
\usepackage{longtable}
\usepackage{amsfonts}
\usepackage{amssymb}
\usepackage{graphicx}
\usepackage{epstopdf}
\usepackage{multirow}
\usepackage{subcaption}
\usepackage{verbatim}
\usepackage{enumitem}
\usepackage{color}
\usepackage{hyperref}
\usepackage{xspace}
\usepackage{slashbox}

\usepackage{float}
\restylefloat{figure}

\usepackage[linesnumbered,ruled]{algorithm2e}

\usepackage{self}

\newtheorem{proposition}{Proposition}
\newtheorem{definition}{Definition}

\newtheorem{example}{Example}

\newtheorem{assumption}{Assumption}

\usepackage{hyperref}
\usepackage{xcolor}
\hypersetup{
	colorlinks,
	linkcolor={red!50!black},
	citecolor={blue!50!black},
	urlcolor={blue!80!black}
}

\usepackage[numbers]{natbib}

\graphicspath{{small_img/}}

\begin{document}

\title{Measuring and reducing the disequilibrium levels of dynamic networks through ride-sourcing vehicle data}

\author{Wei Ma, Zhen (Sean) Qian\\
	Department of Civil and Environmental Engineering\\ Carnegie Mellon University, Pittsburgh, PA 15213\\
	\{weima, seanqian\}@cmu.edu}
\maketitle

\begin{abstract}
Transportation systems are being reshaped by ride sourcing and shared mobility services in recent years. The transportation network companies (TNCs) have been collecting high-granular ride-sourcing vehicle (RV) trajectory data over the past decade, while it is still unclear how the RV data can improve current dynamic network modeling for network traffic management. This paper proposes to statistically estimate network disequilibrium level (NDL), namely to what extent the dynamic user equilibrium (DUE) conditions are deviated in real-world networks. Using the data based on RV trajectories, we present a novel method to estimate the real-world NDL measure. More importantly, we present a method to compute  zone-to-zone travel time data from trajectory-level RV data. This would become a data sharing scheme for TNCs such that, while being used to effectively estimate and reduce NDL, the zone-to-zone data reveals neither personally identifiable information nor trip-level business information if shared with the public. In addition, we present an NDL based traffic management method to perform user optimal routing on a small fraction of vehicles in the network. The NDL measures and NDL-based routing are examined on two real-world large-scale networks: the City of Chengdu with trajectory-level RV data and the City of Pittsburgh with zone-to-zone travel time data. We found that, on weekdays in each city, NDLs are likely high when travel demand is high (thus when congestion is mild or heavy). Generally, a weekend midnight exhibits higher NDLs than a weekday midnight. Many NDL patterns are different between Chengdu and Pittsburgh, which are attributed to unique characteristics of both demand and supply in each city. For instance, NDL in Pittsburgh is much more stable from day to day and from hour to hour, comparing to Chengdu. In addition, we observe that origin-destination pairs with high NDLs are spatially and temporally sparse for both cities. For the Pittsburgh network, we evaluate the effectiveness of NDL-based traffic routing, which shows great potential to reduce total travel time with routing a small fraction of vehicles (1\% in the experiments), even using dated NDL that is estimated in the prior hour.
\end{abstract}
\section{Introduction}
The emerging ride-sourcing services are reshaping the urban transportation systems. Their impacts are two-fold. On one hand, ride-sourcing services added complications to the transportation system, which presents challenges to traditional traffic modeling and management. On the other hand, ride-sourcing services produce massive data ({\em e.g.} vehicle trajectories) that can be potentially leveraged to better understand and manage the whole transportation system. Transportation network companies (TNCs) have been collecting high-granular ride-sourcing vehicle (RV) data over the past decade, while there is a lack of studies on how the RV data can improve current traffic network models and management strategies. In view of this, this paper presents a data-driven approach to connect traffic network models with RV data, and to build an integrated traffic management framework. The main contributions of this study are two-fold: using RV data, 1) this paper evaluates whether the dynamic user equilibrium (DUE) holds and to what extent it is violated in real-world networks; 2) this paper also proposes a holistic framework to manage the traffic network through analyzing ride-sourcing vehicle data. The framework proposes a novel estimation method for network disequilibrium level (NDL) with RV data, a data sharing scheme for TNCs to release aggregated data based on RV without revealing personally identifiable information, and an NDL-based user optimal routing algorithm.



With the boom of smartphones and mobile Internet, on-demand ride-sourcing services are emerging and becoming an indispensable component of urban transportation systems \citep{rayle2016just}. The on-demand ride-sourcing service refers to a service mode that private car owners drive their own vehicles to provide for-hire rides, and the service is usually operated by  transportation network companies (TNCs) such as Uber, Lyft and DiDi chuxing. TNCs match the riders and drivers in real-time and instruct drivers to pick-up and drop-off riders through a real-time routing mechanism. According to a recent report \citep{uber}, Uber alone has covered $551$ cities globally and surpassed two billion rides by July $2016$.

The proliferation of ride-sourcing services has profoundly reshaped of transportation systems and hence stimulated broad discussions and research from various perspectives, including labor market of ride-sourcing drivers and rider behaviors \citep{hall2015analysis, rayle2016just, kooti2017analyzing}, stochastic vehicle dispatching \citep{miao2017data}, pricing strategies \citep{chen2016dynamic, bimpikis2016spatial}, optimal parking provision \citep{xu2017optimal} as well as policies and regulations \citep{ranchordas2015does, zha2016economic, zha2017surge}. RV-related data sets include, but are not limited to:
\begin{enumerate}[label=\roman*)]
	\item {\em Survey data.} The survey data is obtained from surveys. A survey is usually conducted by researchers or public agencies to query riders' or drivers' opinions associated with ride-sourcing services. For example, survey on people's willingness to use ride-sourcing services \citep{dias2017behavioral}, and survey on riders' incentives to switch between Uber and Lyft \citep{di2017switching}.
	\item {\em Zone-to-zone travel time.} Zone-to-zone travel time can be estimated from RV trajectories. For example, Uber Movement provides high-granular zone-to-zone travel time data in many major cities. \citet{pearson11traffic} studied mobility patterns and flow characteristics using Uber Movement data.
	\item {\em Trip-level information:} Trip-level information includes the location of trip origins and destinations, total trip length, and corresponding fares. \citet{zha2017surge} studied the supply-demand relationship on ride-sourcing market using DiDi chuxing trip-level data, and \citet{chen2016dynamic} reveals a positive labor supply elasticity using Uber trip-level data. The trip-level data can also estimate and predict the rider request demand \citep{ke2017short, li2018analysis}.
	\item {\em Trajectory-level information.} Trajectory level information includes detailed second-by-second location of RVs. The main difference between the trajectory-level information and trip-level information is that the trajectories of RVs can be reconstructed with the trajectory-level information, but trip-level information only provides the origins and destinations. \citet{yin2018kalman} estimated vehicle queue length at intersections and \citet{zheng2018traffic} optimized the traffic signal with DiDi chuxing trajectory-level data.
\end{enumerate}

Rich RV datasets present great opportunities to address problems that would be challenging with traditional traffic data and models.  For example, there is a lack of research on whether dynamic user equilibrium (DUE) holds and to what extent it is violated in real-world networks. The dynamic user equilibrium refers to the assumption of network conditions in which no traveler can reduce his/hers disutility (including travel time and fares) by unilaterally changing his/hers route choices (including departure time and route choice) \citep{mahmassani1984dynamic, friesz1993variational, boyce1995solving, huang2002modeling, nie2010solving, friesz2010dynamic}. Essentially, DUE describes how people behave under recurrent traffic conditions.
DUE is often used as the underlying travel behavior model for a wide spectrum of applications including network design, traffic and signal control, origin-destination demand estimation, etc. \citep{josefsson2007sensitivity}. Very few study validates the DUE due to the lack of vehicle trajectory data.

As DUE implies network traffic conditions are precisely known to all travelers and they all make the efforts reaching the optimal routes, many studies challenged the DUE in modeling travelers' behaviors and network conditions in real-world networks \citep{nakayama2001drivers}. \citet{levinson2003value, ben2013impact, nakayama2016effect} showed the accuracy of traffic information provided to travelers has a significant impact on their route selection and the resulting network conditions. \citet{zhu2015people, jan2000using} conducted empirical studies using GPS data and found that most travelers do not follow shortest paths in the network, and some studies indicate that travelers are likely to follow the shortest distance paths \citep{bekhor2006evaluation, prato2006applying} or hyperpaths \citep{ma2013hyperpath}.  Studies also explored alternative models with milder assumptions, such as bounded rationality \citep{mahmassani1987boundedly, di2016boundedly}, statistical traffic equilibrium \citep{nakayama2014consistent, GESTA}, and mean-excess traffic equilibrium \citep{chen2010alpha}.

To the best of our knowledge, none of the previous literature has intensively studied whether and to what extent DUE is violated in real-world networks. In this paper, we discuss and evaluate the concept of network disequilibrium level (NDL) to measure how far real networks are away from DUE, and how NDL varies by location and time.
High NDL may imply an inefficient network in which travelers are unable to choose optimal paths based on their perceived traffic information \citep{ben1991dynamic}, while low NDL implies that DUE is approximately achieved \citep{boyce1995solving}.

The real-time NDL can be used to infer network conditions and information inefficiency in traffic networks and thus can be used for real-time traffic management. The traffic management problem has been extensively studied for decades, readers are referred to \citet{papageorgiou2003review, gao2005optimal, braekers2016vehicle} for a comprehensive review. A number of traffic management studies aim at achieving DUE conditions with in-vehicle communication devices. For example, \citet{wang2001feedback} proposed a feedback and iterative routing scheme to drive the network towards DUE condition, and \citet{du2015coordinated} presented a framework to allow travelers to coordinate with each other to achieve DUE conditions. However, none of the previous studies have studied traffic management strategies for DUE using RV data. 


Data privacy is another important issue of using RV data for transportation management. TNCs are reluctant to share the trip-level and trajectory-level information data, as the traveler's identity and activities might be revealed from the data even with standard de-identification process \citep{de2015unique}. Especially after the Facebook--Cambridge Analytica data scandal \citep{wikifacebook}, IT companies become even more cautious about releasing any information that is potentially related to individuals. Till now, there is a lack of data sharing scheme that allows TNCs to share data based on RV that is 1) aggregated without revealing personally identifiable information, 2) aggregated without revealing sensitive trip-level business information, and 3) proven to be useful to transportation management and planning.

To conclude, the following four questions have not been addressed in the previous studies:
\begin{enumerate}[label=\arabic*)]
	\item What are the unique characteristics of RV data and how do the RV data contribute to traffic network models?
	\item How to evaluate the real-world network disequilibrium level (NDL) in dynamic traffic networks using RV data?
	\item Is there an effective data sharing scheme for TNCs to share aggregated data based on RV that contains neither personally identifiable information nor business sensitive information?
	\item How do the RV data and the NDL estimates contribute to real-world traffic management?
\end{enumerate}

By integrating the traffic network model and RV data, we address the above four questions in this paper. This paper develops a novel theoretical framework for estimating real-world NDLs with trajectory-level or zone-to-zone data based on RVs. A data sharing scheme with zone-to-zone RV data for TNCs is proposed to ensure protecting user privacy and sensitive business information, while the zone-to-zone data still being effective in modeling and managing traffic. This paper also rigorously formulates the traffic management problem with NDL measures and proposes a real-time NDL-based traffic management method to achieve user optimal routing. Finally, we examine the proposed NDL measures on two large-scale real networks using two sources of RV data, respectively. 

The rest of this paper is organized as follows. Section \ref{sec:notation} presents the notations used in this paper. Section~\ref{sec:ndl} discusses the definition and formulation of NDL measure, followed by Section~\ref{sec:ndl2} presenting the estimation of NDL with RV data and the NDL-based traffic management method. In Sections~\ref{sec:didi} and \ref{sec:uber}, two large-scale networks are used to examine the proposed NDL measure and traffic management method. Finally, conclusions are drawn in Section \ref{sec:con}.

\section{Notations}
\label{sec:notation}
All the notations will be introduced in context, and Table~\ref{tab:notation} provides a summary of the basic notations to be referred in the reminder of this paper.

\begin{longtable}{p{2cm}p{13cm}}
	\caption{\footnotesize List of notations}
	\label{tab:notation}
	\endfirsthead
	\endhead
	\multicolumn{2}{c}{\textbf{Network Variables}}\\
	$K_q$ & The set of all origin-destination (OD) pairs\\
	$K_{rs}$ & The set of all paths between OD pair $rs$\\
	
	
	\multicolumn{2}{c}{}\\
	\multicolumn{2}{c}{\textbf{Network flow related variables}}\\
	$t$ & The departure time of path flow or OD flow, and $t$ can be either time point in continuous time space or time interval in discrete time space\\
	$T$ & The set of all possible departure time from all path and OD flow\\
	$F_{rs}^{k}(t)$ & The $k$th path flow for OD pair $rs$ departing at time $t$\\
	$Q_{rs}(t)$ & The flow of OD pair $rs$ departing at time $t$\\
	$C_{rs}^{k}(t)$ & The path cost for path $k$ for OD pair $rs$ departing at time $t$\\
	$C_{rs}(t)$ & The average OD cost for OD pair $rs$ departing at time $t$\\
	$p_{rs}^{k}(t)$ & The route choice portion of choosing path $k$ in all paths between OD pair $rs$ at time $t$\\
	\multicolumn{2}{c}{}\\
	\multicolumn{2}{c}{\textbf{RV related variables}}\\
	$N$ & The set of all orders of ride-sourcing service  (orders will be defined in section~\ref{sec:ndl2})\\
	$v_{i}$ & Trajectory record of $i$th order\\
	$N_{rs}(t)$ & The set of all orders departing from $r$ to $s$ at time $t$\\
	$D_{rs}(t)$ & The network disequilibrium level for OD pair $r-s$ departing at time $t$\\
	$\gamma_{i}$ & The travel time to complete the trip for ride-sourcing vehicle $v_i$\\
	$\kappa_{j}$ & The new travel time to complete the trip for the $j$th zone-to-zone ride-sourcing vehicle record\\
	$\tilde{N}_{rs}(t)$ & The set of all zone-to-zone ride-sourcing vehicle records departing from $r$ to $s$ at time $t$\\
	$\tilde{C}_{rs}(t)$ & The empirical cost for path $k$ for OD pair $rs$ departing at time $t$ estimated by zone-to-zone ride-sourcing vehicle records\\
	$\tilde{C}_{r \tau s}(t)$ & The alternative path cost for path $k$ for OD pair $rs$ departing at time $t$ using $\tau$ as connecting points\\
\end{longtable}

\section{Modeling the network disequilibrium level}
\label{sec:ndl}

In this section, we briefly review the concepts of dynamic user equilibrium (DUE). We discuss the violation of DUE conditions and its consequences. The concept of network disequilibrium level (NDL) is proposed to quantitatively measure to what extent a network deviates from DUE. We will also discuss difficulties to estimate NDL using traditional traffic data.



The Dynamic User Equilibrium (DUE) is presented in Definition~\ref{df:due}.

\begin{definition}[Dynamic User Equilibrium]
	\label{df:due}
	In a road network $\G = (V,E)$, where $V$ and $E$ represent all the intersections and directed road segments in the directed graph $\G$. We denote $F_{rs}^{k}(t)$ as the path flow departing at time $t$ for $k$th path of origin-destination (OD) pair $rs$, and $C_{rs}^{k}(t)$ as the  travel time to complete the trip for the flow of $k$th path of OD $rs$ departing at time $t$. We note the time index $t$ can represent either a time point or a time interval in this paper. The network is in dynamic user equilibrium (DUE) if Equation~\ref{eq:due} holds.
	
	\begin{eqnarray}
	\label{eq:due}
	F_{rs}^{k}(t)\left(C_{rs}^{k}(t) - \Pi_{rs}(t) \right) = 0, \forall k \in K_{rs}, rs \in K_q, t \in T
	\end{eqnarray}
	where
	\begin{eqnarray}
	\Pi_{rs}(t) &=& \min_{k \in K_{rs}} C_{rs}^{k}(t), \forall rs \in K_q, t \in T\\
	F_{rs}^k(t) &\geq& 0, \forall k \in K_{rs}, rs \in K_q, t \in T
	\end{eqnarray}
\end{definition}

According to the direct interpretation of Definition~\ref{df:due}, DUE represents the traffic conditions in which no travelers can unilaterally change their routes to achieve earlier arrival time. Travelers do not choose the routes that exceed the minimal travel time, and only the routes with minimal travel time $\Pi_{rs}(t)$ are chosen.



We can illustrate DUE from the perspective of perfect information assumption. DUE assumes that travelers have perfect information about network conditions, and they always choose the user optimal path in the network. To be precise, there are two different concepts involved: 1) A travelers has incentives to choose a route that is the best one given his/her limited knowledge and efforts. We view that route he/she takes as the actual path; 2) if the travel time of all roads at any particular time of day were known, we can compute the optimal path based on the actual time-varying travel time. This path is viewed as a user optimal path. We note the actual path taken by travelers can base on travelers' experiences, traffic information (e.g. Google Maps) or en-route decisions, while the user optimal path is determined by the actual road travel time. One can clearly see when DUE holds, the actual paths and user optimal paths are the same. In contrast, the network disequilibrium level can be measured by the difference between the actual path taken by travelers and the user optimal path.





In the following sections, we will discuss how the network disequilibrium level is quantitatively formulated and computed.
First of all, we formally discuss the general concept of network disequilibrium level.

\subsection{Network disequilibrium level (NDL)}

Any measure of network disequilibrium level should satisfy Definition~\ref{df:dis}.
\begin{definition}
	\label{df:dis}
	Denote $D_{rs}(t)$ as the network disequilibrium level (NDL) for OD pair $rs$ departing at time $t$, then $D_{rs}(t)$ satisfies the following property:
	
	\begin{itemize}
		\item $D_{rs}(t) \in [0, \infty)$
		\item $D_{rs}(t) = 0$ when Equation~\ref{eq:due} holds for $rs \in K_{q}$ and $t \in T$.
	\end{itemize}
\end{definition}

We first state two different definitions of NDL in Definition~\ref{eq:dis11} and Definition~\ref{eq:dis1}, both of which stem from the merit function in DUE formulations \citep{lu2009equivalent, nie2010solving}.

\begin{definition}[Merit function-based NDL]
	The merit function-based NDL is presented in Equation~\ref{eq:dis11}.
	\label{def:dis11}
	\begin{eqnarray}
	\label{eq:dis11}
	D_{rs}^{\mathcal{M}}(t) = \sum_{k \in K_{rs}} F_{rs}^{k}(t)\left(  C_{rs}^{k}(t) - \Pi_{rs}(t) \right)
	\end{eqnarray}
\end{definition}

Definition~\ref{def:dis11} is a direct adaptation of Equation~\ref{df:due}. It is also possible to measure the NDL by the percentage deviation of the minimal path travel time. For example, Equation~\ref{eq:dis11} can be normalized by dividing $\Pi_{rs}(t)$. In this paper, we focus on the original definition in Equation~\ref{eq:dis11}. We note the path flow can be computed through the route choice portion $p_{rs}^{k}(t)$, the portion of flow choosing path $k$ in all paths between OD pair $rs$ at time $t$, as presented in Equation~\ref{eq:route}.

\begin{eqnarray}
\label{eq:route}
F_{rs}^{k}(t) = p_{rs}^{k}(t)Q_{rs}(t)
\end{eqnarray}

Definition~\ref{def:dis1} is a flow-scale-free version of Definition~\ref{def:dis11}, meaning the NDL is independent of the scale of path flow.

\begin{definition}[Flow-scale-free NDL]
	\label{def:dis1}
	The flow-scale-free NDL is presented in Equation~\ref{eq:dis1}.
	\begin{eqnarray}
	\label{eq:dis1}
	D_{rs}^{\mathcal{F}}(t) = \sum_{k \in K_{rs}} p_{rs}^{k}(t)\left(  C_{rs}^{k}(t) - \Pi_{rs}(t) \right)
	\end{eqnarray}
\end{definition}

The merit function-based NDL can be viewed as the total NDL, while the flow-scale-free NDL can also be viewed as the average of total NDL for each traveler in OD pair $rs$, as shown in Equation~\ref{eq:rela}. 

\begin{eqnarray}
\label{eq:rela}
D_{rs}^{\mathcal{M}}(t) = Q_{rs}(t)D_{rs}^{\mathcal{F}}(t), \forall rs \in K_q, t \in T
\end{eqnarray}

We further define the origin-based NDL and destination-based NDL by aggregating NDLs of all OD pairs to origins and destinations, as presented in Equation~\ref{eq:odndl1} and~\ref{eq:odndl2}.

\begin{eqnarray}
\label{eq:odndl1}
D_{r}(t)&=& \sum_{s} D_{rs}(t)\\
\label{eq:odndl2}
D_{s}(t) &=& \sum_{r} D_{rs}(t)
\end{eqnarray}

In this paper, we assume that the travelers' ``true'' disutility is represented by travel time \citep{papinski2009exploring}, while other factors such as distances, road tolls, route preferences are not considered \citep{zhu2015people}. The definition of NDL can be extended to include the generalized disutility of travelers, while additional information might be required to compute the NDL.
The NDL is attributed to travelers' inability to know what the shortest-time route is. Many other possible causes of NDL exist \citep{mahmassani1987boundedly, shao2006reliability, di2016boundedly, GESTA, yu2018day}. This paper will focus on evaluating and reducing the NDL first, while exploring the causes of NDL will be left for future research.

Evaluating NDL requires full knowledge of path flow.
However, obtaining time dependent path flow (or route choice) is notoriously difficult in large-scale networks given limited number of surveillance sensors \citep{lu2013dynamic}.
In addition, estimation methods for dynamic path flow (or dynamic route choice) are usually based on the assumption that networks are in DUE. Therefore, computing NDLs through Equation~\ref{eq:dis1} and Equation~\ref{eq:dis11} becomes implausible on large-scale networks. To have a plausible definition of the NDL, characteristics of ride-sourcing vehicle (RV) data will be discussed and utilized. 
In the following sections, we present two ways to measure the network disequilibrium levels with either trajectory-level (Section~\ref{sec:disagg}) or zone-to-zone (Section~\ref{sec:agg}) RV data.

\section{Estimating NDL with ride-sourcing vehicle data}
\label{sec:ndl2}
Section~\ref{sec:ndl} has defined and discussed the concept of NDL. This section focuses on how to estimate the NDL measures in the real-world using RV data only. We first discuss the characteristics of RV data, then the trajectory-level RV data is used to estimate the NDL. Due to the privacy concerns for the trajectory-level data, we propose a data sharing scheme to compute and share RV data that are aggregated to zone-to-zone. A method to estimate the NDL with the zone-to-zone RV data is further proposed. Lastly, we build a real-time traffic management framework for user optimal routing based on NDL.

\subsection{Ride-sourcing vehicle data}
\label{sec:ride}

In this sub-section, we discuss the characteristics of RV data. We define an ``order'' as one transaction of a trip completed by a pair of a driver and a rider, which starts from the pick-up of the rider and ends with the drop-off of the rider. The cruising process and the process between accepting the rider request and picking up the rider are not considered, since the behaviors of RVs during these two processes may not follow the pattern of a trip generated by travelers. For example, Uber allows drivers to accept a new rider request before finishing current order, hence the trajectory of RV between accepting the order and picking up the riders can be very random. For the behaviors of RVs during an order, we argue Assumption~\ref{as:uni} holds.

\begin{assumption}[Uniform sampling]
	\label{as:uni}
	The ride-sourcing vehicles from origin $r$ to destination $s$ departing at time $t$ are uniformly distributed among all vehicles that depart from origin $r$ to destination $s$ at time $t$. In other words, the route choice behaviors of RVs are the same as conventional non-sharing vehicles. The probability of an RV departing from $r$ to $s$ at $t$ to choose path $k$ equals to $p_{rs}^{k}(t)$, which is the route choice probability of all vehicles departing from $r$ to $s$ at $t$.
\end{assumption}

Assumption~\ref{as:uni} claims that the route choice behaviors of RVs are the same as the conventional vehicles, hence the RV trajectories are uniformly sampled from all trajectories in the network.
Assumption~\ref{as:uni} is attributed to that RV trips may approximate personal vehicle trips. Both the drivers and riders of RVs are more likely to exhibit a wide spectrum of socio-demographics and driving behavior that are close to non-sharing trips than other biased probe samples (such as taxis or trucks). We note the market penetration rate of RVs for different OD pairs and departure times can be different as long as the RVs are uniformly distributed among all vehicles.

The behaviors of trucks usually do not follow Assumption~\ref{as:uni}. Trucks usually operate on highways and major roads due to the prohibitions and road limitations in urban areas, and truck drivers usually follow fixed routes. The objectives of truck drivers are to deliver goods on time, and some trucks may not be incentivized to arrive as early as possible. Instead, they may prefer behaviors that save fuel use.


The behaviors of taxi drivers do not follow Assumption~\ref{as:uni} either. Taxi drivers usually work full-time, and they generally have more driving experiences than an average traveler \citep{shi2014survey}. Therefore, using taxi data to represent the average drivers might introduce a substantial bias. On the contrast, ride-sourcing vehicle drivers can work either full-time or part-time, and their socio-demographics may be more representative (thus less biased) \citep{hall2018analysis}. 

Having randomly sampled personal vehicle data would be ideal. 
However, there is a significant portion of the personal vehicles with no GPS or data transmission device installed, which also introduce a sampling bias. In addition, tracking personal vehicles is controversial, as most of the Americans believe this kind of tracking is extremely invasive \citep{cnn}. Till now, there is no public dataset available for personal vehicle trajectories.






The unique characteristics of RVs enable them being a possibly representative trip sample in the complex real-world network. Though drivers do their best to choose the optimal route, their actual paths may not turn out to be the optimal paths. This fact can be utilized to construct a measure of network disequilibrium levels, and details will be discussed in the following sections. We will mathematically show that a very small fraction of RVs can actually approximate disequilibrium measures.

\subsection{Estimating NDL with trajectory-level RV data}
\label{sec:disagg}

In this sub-section, we estimate NDL with trajectory-level RV data. We note that the trip-level RV data can also be used to estimate NDL. We focus on the trajectory-level RV data from now on.

We denote an order as $v_{i}$, where $i$ is the index of order, and $i \in N$ where $N$ denotes the set of all orders. During each order, the RV sends its location information to TNCs frequently. Each order contains a sequence of location and time information (which are referred as messages), as presented in Equation~\ref{eq:trajectory}.
\begin{eqnarray}
\label{eq:trajectory}
v_{i} = \{(t_{i}^{0}, l_{i}^{0}), (t_{i}^1, l_{i}^1), (t_{i}^2, l_{i}^2), \cdots, (t_{i}^{T_{i}-1}, l_{i}^{T_{i}-1}) \}
\end{eqnarray}
where $T_i$ is denoted as the number of messages in order $i$. $t_i^h$ denotes the $h$th time stamp when the information is sent, and $l_i^h$ is the location of the RV at time $t_i^h$. We divide all the orders into different subsets, and $N_{rs}(t)$ denotes the set of orders that departs from $r$ at time $t$ and arrives at $s$, as presented in Equation~\ref{eq:match}.
\begin{eqnarray}
\label{eq:match}
N_{rs}(t) = \left\{i \in N | t_{i}^0 = t , \xi(l_{i}^0) = r ,  \xi\left(l_{i}^{T_{i}-1}   \right) =s \right\}
\end{eqnarray}
where $\xi(\cdot)$ matches the location to its corresponding traffic analysis zone (TAZ). We further denote $\gamma_{i}$ as the travel time to complete the RV trip $i$, and it can be computed by Equation~\ref{eq:gamma}.

\begin{eqnarray}
\label{eq:gamma}
\gamma_{i} =  t_{i}^{T_{i}-1} - t_{i}^0
\end{eqnarray}

For each OD pair $rs$, $\gamma_i, i \in N_{rs}(t)$ is a sample of path travel time $C_{rs}^k(t)$ if the RV trip $v_{i}$ takes route $k$ and departs at time $t$. Based on Assumption~\ref{as:uni}, the probability of $v_i$ being on path $k$ is $p_{rs}^k(t)$. 
We are now ready to present an estimator for NDL using trajectory-level RV data in Equation~\ref{eq:dis2}.

\begin{eqnarray}
\label{eq:dis2}
D_{rs}^{\mathcal{D}}(t) = \frac{1}{|N_{rs}(t)|} \sum_{i \in N_{rs}(t)}  \gamma_{i} - \min_{i \in N_{rs}(t)} \gamma_{i}
\end{eqnarray}

First we prove that $D_{rs}^{\mathcal{D}}(t)$ approximates $D_{rs}^{\mathcal{F}}$ when the number of sampled orders is sufficiently large, as presented in Proposition~\ref{prop:dis1}.

\begin{proposition}
	\label{prop:dis1}
	When $|N_{rs}(t)| \to \infty, \forall rs \in K_{q}, t \in T$, we have
	\begin{eqnarray}
	D_{rs}^{\mathcal{D}}(t)  \xrightarrow{P} D_{rs}^{\mathcal{F}}(t)
	\end{eqnarray}
\end{proposition}

\begin{proof}
	We compute the expectation of $\gamma_{i}$ in each subset $N_{rs}(t)$ based on Assumption~\ref{as:uni},
	\begin{eqnarray}
	\mathbb{E}_{i}\left[\gamma_{i}\right] = \sum_{k \in K_{rs}}p_{rs}^{k}(t) C_{rs}^{k}(t), \forall i \in N_{rs}(t)
	\end{eqnarray}
	By Law of Large Numbers (LLN), we have
	\begin{eqnarray}
	\frac{1}{|N_{rs}(t)|} \sum_{i \in N_{rs}(t)}  \gamma^{i} \xrightarrow{P} \sum_{k \in K_{rs}} p_{rs}^{k}(t) C_{rs}^{k}(t), \forall i \in N_{rs}(t)
	\end{eqnarray}
	Then the cumulative distribution function  $\mathbb{F}_{\cdot}(y)$ of $\min_{i \in N_{rs}(t)} \gamma_i$ is presented in Equation~\ref{eq:mintt}.
	\begin{eqnarray}
	\label{eq:mintt}
	\mathbb{F}_{\min_{i \in N_{rs}(t)} \gamma_i}\left(y\right) = 1 - \prod_{i \in N_{rs}(t)}\left( 1 -  \mathbb{F}_{\gamma_{i}}(y)\right)
	\end{eqnarray}
	where
	\begin{eqnarray}
	\mathbb{F}_{\gamma_i}(y)  \begin{cases}
	=0 &\text{if $y \leq \Pi_{rs}(t)$ and $i \in N_{rs}(t)$ }\\
	> 0 & \text{else}
	\end{cases}
	\end{eqnarray}
	When $|N_{rs}(t)| \to \infty$, we have
	\begin{eqnarray}
	\mathbb{F}_{\min_{i \in N_{rs}(t)} \gamma_i}\left(y\right) =
	\begin{cases}
	0 &\text{if $y \leq \Pi_{rs}(t)$}\\
	1 & \text{else}
	\end{cases}
	\end{eqnarray}
	Hence,
	\begin{eqnarray}
	\min_{i \in N_{rs}(t)} \gamma_i \xrightarrow{P} \Pi_{rs}(t)
	\end{eqnarray}
	Therefore Proposition~\ref{prop:dis1} holds based on continuous mapping theorem.
\end{proof}

Secondly, we show that $D_{rs}^{\mathcal{D}}(t)$ is a biased estimator with finite data, as presented in Proposition~\ref{prop:biased}.

\begin{proposition}[Biased NDL estimator]
	\label{prop:biased}
	If $|N_{rs}(t)| < \infty$, then
	\begin{eqnarray}
	\mathbb{E} \left[D_{rs}^{\mathcal{D}} (t)\right] \geq D_{rs}^{\mathcal{F}}(t)
	\end{eqnarray}
\end{proposition}
\begin{proof}
	Since the minimum function is convex, we have Equation~\ref{eq:jensen} holds by Jensen's inequality.
	\begin{eqnarray}
	\label{eq:jensen}
	\mathbb{E} \left[\min_{i \in N_{rs}(t)} \gamma_{i} \right] \leq  \min_{k \in K_{rs}}  C_{rs}^k(t)
	\end{eqnarray}
	Further by Proposition~\ref{prop:dis1}, we have $\mathbb{E}_{i}\left[\gamma_{i}\right] = \sum_{k \in K_{rs}}p_{rs}^{k}(t) C_{rs}^{k}(t), \forall i \in N_{rs}(t)$. Then we can get
	\begin{eqnarray}
	\mathbb{E} \left[D_{rs}^{\mathcal{D}} \right] &=& \mathbb{E} \left[ \frac{1}{|N_{rs}(t)|} \sum_{i \in N_{rs}(t)}  \gamma_{i}\right] - \mathbb{E} \left[\min_{i \in N_{rs}(t)} \gamma_{i} \right] \\
	&\geq&\sum_{k \in K_{rs}}p_{rs}^{k}(t) C_{rs}^{k}(t)  - \min_{k \in K_{rs}} C_{rs}^{k}(t)\\
	&=&D_{rs}^{\mathcal{F}}
	\end{eqnarray}
	
\end{proof}

As can be seen from Proposition~\ref{prop:dis1}, the estimator $D_{rs}^{\mathcal{D}}(t)$ approximates $D_{rs}^{\mathcal{F}}(t)$ when $|N_{rs}(t)| \to \infty$. However, there are two issues regarding to the NDL measure estimator $D_{rs}^{\mathcal{D}}(t)$: 1) based on Proposition~\ref{prop:biased}, Estimating NDL may not be accurate when $N_{rs}(t)$ is small. This is oftentimes the case given a very small penetration rate of RVs currently; 2) the trajectory-level (or trip-level) data may reveal personally identifiable information, and thus TNCs will not share the trajectory (or trip-level) data to public sectors. Hence the trajectory-level data is extremely difficult to obtain.

To further address these issues, we present a new method to estimate NDL using zone-to-zone travel time. The new method relies only on the zone-to-zone travel time of RV data, and it potentially requires fewer data since the travel time is estimated by re-using the trajectory data. In addition, the trajectory data is aggregated to provide zone-level travel time information only, which does not contain personally identifiable information.

\subsection{Estimating NDL with zone-to-zone aggregated RV data}
\label{sec:agg}

In this sub-section, we first present a method to compute the zone-to-zone travel time using RV data. An unimode testing is presented to validate the zone-to-zone travel time. We also provide a scheme for TNCs to release RV data without revealing personally identifiable information. Finally the method to estimate NDL using the zone-to-zone travel time will be presented.

\subsubsection{Zone-to-zone travel time}

We first denote $\phi_{i}$ as the set of all tuples such that each tuple represents a pair of intermediate points along the RV trajectory $v_i$.
$\phi_{i}$ can be computed by Equation~\ref{eq:tuples}.

\begin{eqnarray}
\label{eq:tuples}
\phi_{i} = \left\{ (e, e') | 0 \leq e < e' < T_{i} , e \in \mathbb{Z}, e' \in \mathbb{Z}\right\}
\end{eqnarray}
where $\mathbb{Z}$ represents the set of integers. Each tuple $(e,e')$ represents a segment of the whole trajectory of $v_{i}$, and $e, e'$ represent the starting point and ending point of the segment, respectively. We compute the travel time $\kappa_{ij}$ for each segment, as presented in Equation~\ref{eq:kappa}.

\begin{eqnarray}
\label{eq:kappa}
\kappa_{ij} =  t_{i}^{\phi_{i}(j, 1)} - t_{i}^{\phi_{i}(j, 0)}
\end{eqnarray}
where $\phi_{i}(j, 0)$ and $\phi_{i}(j, 1)$ denote the first and second element in the $j$th tuple in set $\phi_{i}$. $\kappa_{ij}$ can be viewed as a new trajectory travel time sample from $l_{i}^{\phi_{i}(j, 0)}$  to $l_{i}^{\phi_{i}(j, 1)}$. We divide the new trajectory travel time into different sets by its origin and destination, as presented in Equation~\ref{eq:sets2}.

\begin{eqnarray}
\label{eq:sets2}
\tilde{N}_{rs}^{i}(t) = \left\{j \in \Phi| t_{i}^{\phi_{i}(j, 0)} = t, \xi\left(l_{i}^{\phi_{i}(j, 0)}\right) = r, \xi\left(l_{i}^{\phi_{i}(j, 1)}\right) = s\right\}
\end{eqnarray}

To simplify the notation, we further denote
\begin{eqnarray}
\tilde{N}_{rs}(t) = \cup_{i \in N} \tilde{N}_{rs}^{i}(t)
\end{eqnarray}
where the trajectories of all orders are used to form a much richer data set $\tilde{N}_{rs}(t)$. Within $\tilde{N}_{rs}(t)$, we can drop the index $i$ in $\kappa_{ij}$ and re-index the $j$ to make $\kappa_{j}$. Now $j$ is the element index in set $\tilde{N}_{rs}(t)$ instead of $\tilde{N}_{rs}^{i}(t)$. $j$ represents a segment of trajectories from all orders in the data set, and thus can be seen irrelevant to a particular order in $\tilde{N}_{rs}(t)$.

We define $\tilde{C}_{rs}(t)$ as the average sample (or observed) travel time for OD $rs$ departing at time $t$. For each $\tilde{N}_{rs}(t)$, $\kappa_{j}, j \in \tilde{N}_{rs}(t)$  can be viewed as a sample of $\tilde{C}_{rs}(t)$. Hence,  $\tilde{C}_{rs}(t)$ can be computed by Equation~\ref{eq:ces}.
\begin{eqnarray}
\label{eq:ces}
\tilde{C}_{rs}(t) = \frac{1}{|\tilde{N}_{rs}(t)|} \sum_{j \in\tilde{N}_{rs}(t)} \kappa_{j}, \forall  rs \in K_q, t\in T
\end{eqnarray}
By using the above-discussed approach, TNCs will not share the trajectory-level information, instead they can simply share statistics related to $\tilde{C}_{rs}(t)$ ({\em e.g.} Uber Movement data). The zone-to-zone travel time $\tilde{C}_{rs}(t)$ is aggregated over the zones, containing neither personally identifiable information nor any trip information. 


\begin{example}[Processing trajectory data]
	\label{ex:1}
	In this example, we will demonstrate how to process the trajectory data to get the zone-to-zone travel time information. The map and two RV trajectories are presented in Figure~\ref{fig:e1}. Each cell in the map represents one traffic analysis zone. We use the x-coordinate and y-coordinate to determine the location of the zone we are referring to, for example $(1,2)$ denotes the zone in which point $(t_{0}^1, l_{0}^1)$ is located. We also assume $t_{0}^2 = t_{1}^0$.
	
	\begin{figure}[h]
		\centering
		\includegraphics[scale = 0.7]{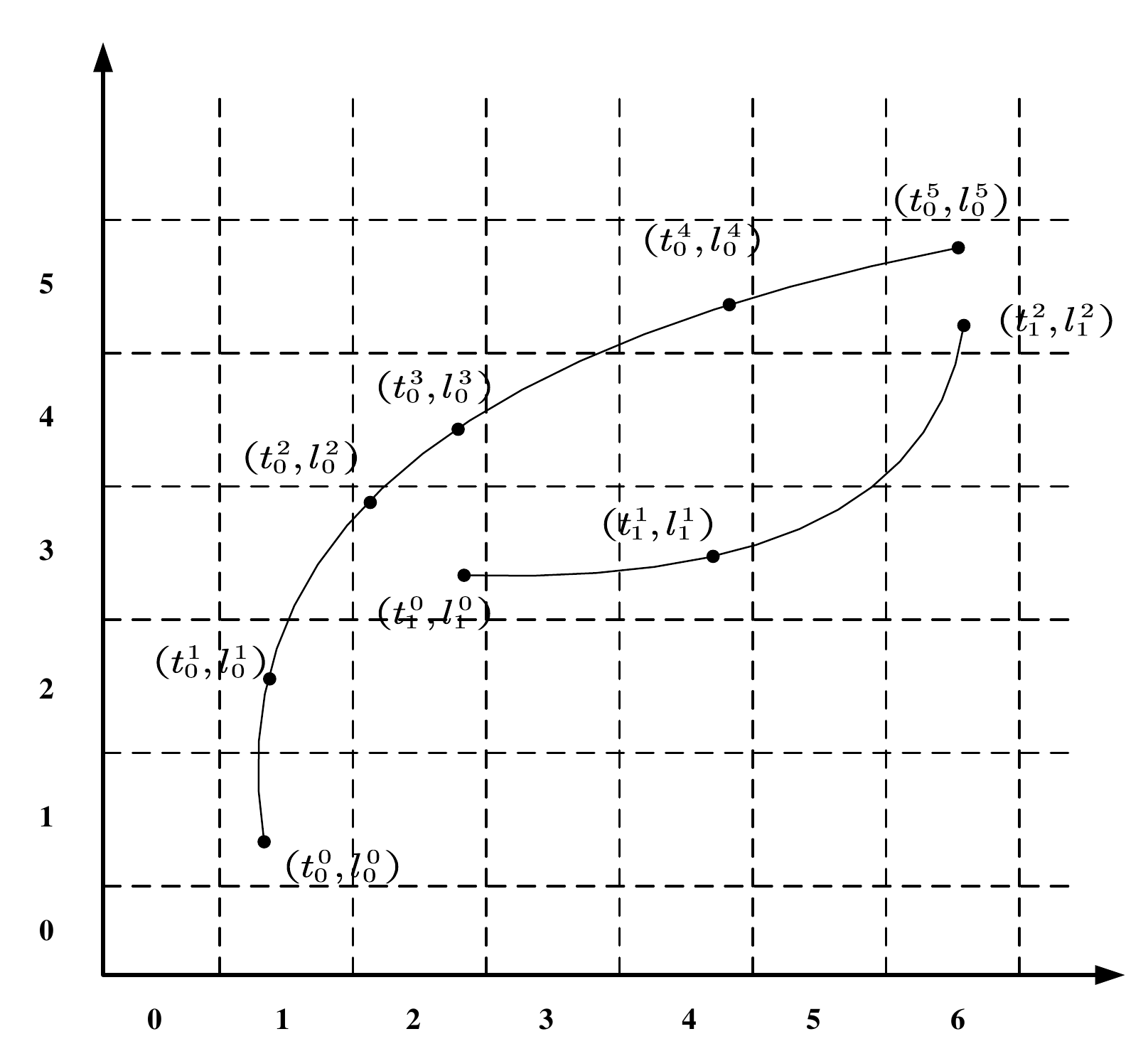}
		\caption{Example of two RV trajectories}
		\label{fig:e1}
	\end{figure}	
	
	For both trajectories, we have
	\begin{eqnarray}
	\phi_{0} &=& \{(0,1), (0,2), (0,3), (0,4), (0, 5), (1,2), (1,3), (1,4), (1, 5), (2,3), (2,4), (2, 5), (3,4), (3, 5), (4,5)\}\\
	\phi_{1} &=& \{(0,1), (0,2), (1,2)\}
	\end{eqnarray}
	
	Then the segmented travel time is computed as follows:
	\begin{eqnarray}
	\kappa_{00} = t_{0}^1 - t_{0}^0\quad \kappa_{01} = t_{0}^2 - t_{0}^1\quad\kappa_{02} = t_{0}^3 - t_{0}^0\quad\kappa_{03} = t_{0}^4 - t_{0}^0\quad\kappa_{04} = t_{0}^5 - t_{0}^0\\
	\kappa_{06} = t_{0}^2 - t_{0}^1\quad \kappa_{07} = t_{0}^3 - t_{0}^1\quad\kappa_{08} = t_{0}^4 - t_{0}^1\quad\kappa_{09} = t_{0}^5 - t_{0}^1\quad\kappa_{0(10)} = t_{0}^3 - t_{0}^2\\
	\kappa_{0(11)} = t_{0}^4 - t_{0}^2\quad \kappa_{0(12)} = t_{0}^5 - t_{0}^2\quad\kappa_{0(13)} = t_{0}^4 - t_{0}^3\quad\kappa_{0(14)} = t_{0}^5 - t_{0}^3\quad\kappa_{0(15)} = t_{0}^5 - t_{0}^4\\
	\kappa_{1(15)} = t_{1}^1 - t_{1}^0\quad \kappa_{1(16)} = t_{1}^2 - t_{1}^0\quad\kappa_{1(17)} = t_{1}^2 - t_{1}^1
	\end{eqnarray}
	
	We then omit index $i$ in the $\kappa_{ij}$ and re-index to get $\kappa_{0}, \cdots, \kappa_{j}, \cdots, \kappa_{17}$. Then the index set of each OD pair $\tilde{N}_{rs}(t)$ is computed as follows:
	\begin{eqnarray}
	\tilde{N}_{(1,1)  (1,2)}(t_{0}^0) &=& \{0\}\\
	\tilde{N}_{(1,1)  (2,3)}(t_{0}^0) &=& \{1\}\\
	&\vdots& \nonumber\\
	\tilde{N}_{(2,3)  (6,5)}(t_{0}^2) &=& \{11, 16\}\\
	&\vdots&\nonumber
	\end{eqnarray}
	
	Finally, the segmented travel time $\kappa_{j}$ can be used to estimate  $\tilde{C}_{rs} (t)$ between any OD pair. 
	
\end{example}


\subsubsection{Unimodality test}
When converting the trajectory data to zone-to-zone travel time data, we divide the new trajectory travel time $\kappa_j$ into different sets $\tilde{N}_{rs}(t)$ by its origin and destination, while the travel time within the origin and destination zones is ignored, as demonstrated in Example~\ref{ex:1}.
In fact, the travel time within the zone can be very large if the traffic analysis zones are not partitioned properly. For example, if there is a highway and a local road in region $A$. The highway passes by region $A$ and there is no exiting ramp connecting the highway and region A. Hence vehicles on the highway have to drive out of region $A$ to get off the highway and then come back to cell $A$ through local roads. In this case, the travel time of RV within cell $A$ can not be ignored. We still use Example~\ref{ex:1} as an example, if $l_2^0$ is on the highway while $l_1^0$ is on the local road (suppose cell $A$ is zone $(2,3)$). Based on the discussions above, $\kappa_{11}$ and $\kappa_{16}$ can be very different. Therefore the estimation of of $C_{(2,3)  (6,5)}(t_0^2)$ using $\kappa_{11}$ and $\kappa_{16}$ is not accurate.

We observe that the distribution of travel time from one zone to other zones usually contains multiple modes when the network topology or road properties are substantially heterogeneous for that zone. To ensure the estimation accuracy of the zone-to-zone travel time, we adopt a unimodality testing method called dip test \citep{hartigan1985dip}. The dip test is run for $\kappa_j, j\in \tilde{N}_{rs}(t)$ to check if the travel time is uni-modal. 
If the dip test fails, we would need to further partition the cell into smaller zones and re-compute the zone-to-zone travel time for the smaller zones. An example of the dip test and region segmentation is further provided in Example~\ref{ex:2}.

\begin{example}
	\label{ex:2}
	In this example, we describe how the dip test is conducted on a simplified network. We consider a network with one OD pair, as presented in the upper part of Figure~\ref{fig:ex2}. There is a highway and a local road connecting the origin $r$ and destination $s$.
	\begin{figure}[h]
		\centering
		\includegraphics[scale = 0.5]{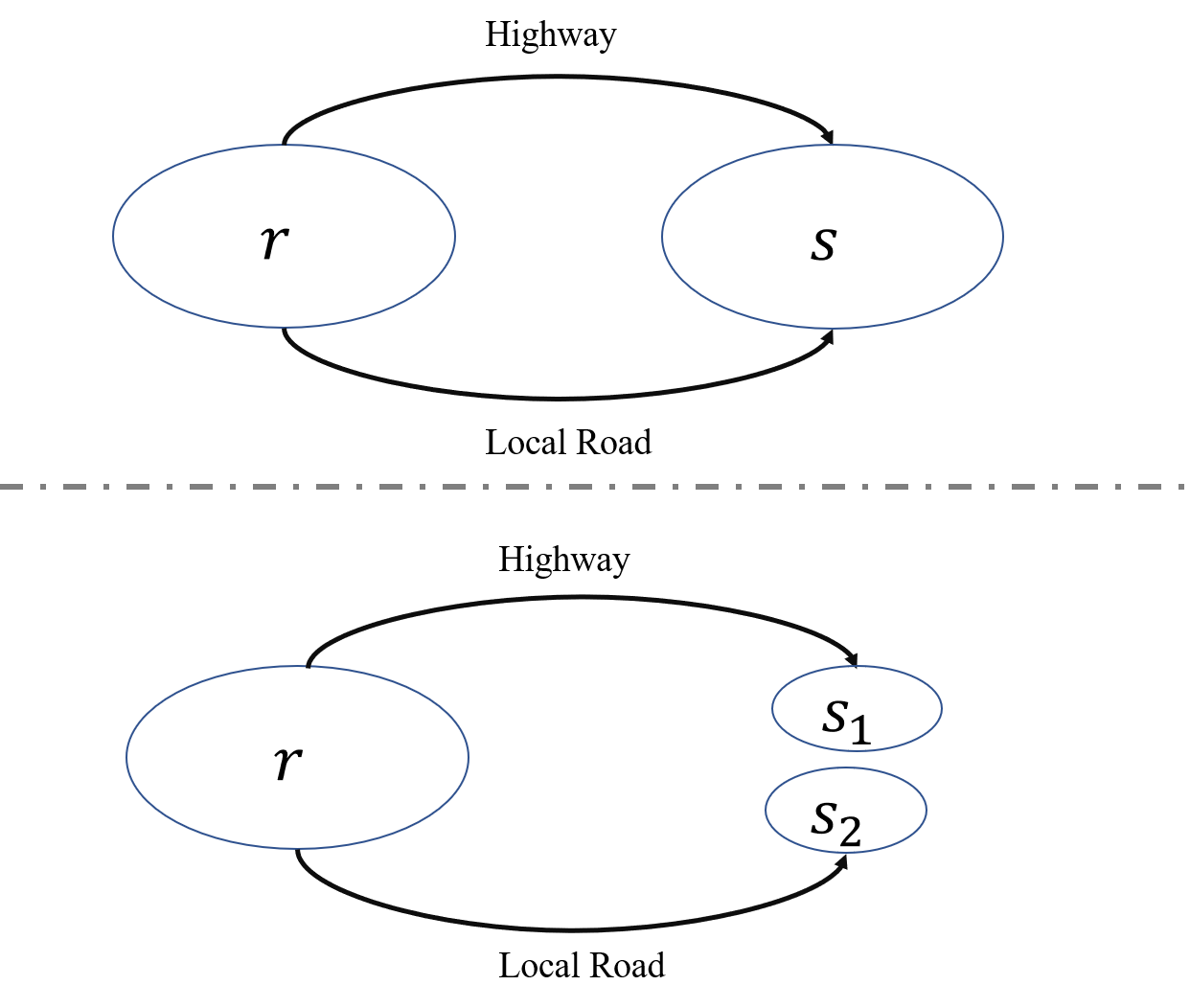}
		\caption{An example of dip test and zone segmentation}
		\label{fig:ex2}
	\end{figure}
	
	Suppose we have ten RV trajectory data, five of them indicate the travel time between $rs$ is $10$ minutes, while the other five trajectories show a $5$ minutes travel time. To be precise, we have $|\tilde{N}_{rs}(t)| = 10$ for a certain $rst$, and $\kappa_j = 5,5,5,5,5,10,10,10,10,10$, respectively. We ran the dip-test and found the p-value is zero, which indicates the travel time is clearly not unimodal.

	The difference in the travel time is because five drivers took the highway and the other five drivers took the local road, To eliminate this discrepancy, we can segment the destination zone $s$ into two zones. One zone $s_1$ contains the end of highway while the other zone $s_2$ contains the end of the local road, as presented in the lower part of Figure~\ref{fig:ex2}. We re-ran the dip-test, the p values are now $0.9029$ for both $s_1$ and $s_2$, which suggest the travel time is unimodal, respectively. By conducting the unimodality test and zone segmentation iteratively, we can eliminate the multi-modal issue appeared in the travel time data.
	
\end{example}

\subsubsection{Measuring the network disequilibrium level}

We present the estimation method of NDL using $\tilde{C}_{rs}(t)$. Before that, we first review the definition of the segmented trajectory $j \in \tilde{N}_{rs}(t)$. Each actual trajectory $i$ corresponds to multiple segmented trajectories, and we can view this as that there are multiple virtual RV drivers on the roads and they depart from $r$ to $s$ at time $t$. Each virtual RV driver is experiencing the travel time $\kappa_j$. The set $\tilde{N}_{rs}(t)$ can be viewed as an enlarged trajectory set driven by the virtual drivers. Based on this observation, we further claim Assumption~\ref{as:nop} is true.

\begin{assumption}
	\label{as:nop}
	The route choice behaviors of the virtual RV drivers are the same as the actual RV drivers.
\end{assumption}

Assumption~\ref{as:nop} implies that the drivers' route choice behavior will not change even if they only drive part of the route. An easy way to think of Assumption~\ref{as:nop} is that when one takes a RV and the route (which is selected by drivers) is displayed on one's phone, if one adds a stopping point along the provided route, Assumption~\ref{as:nop} claims that the updated route will not change.  Assumption~\ref{as:nop} is stronger than Assumption~\ref{as:uni} because it constrains on each segment of a trajectory, while Assumption~\ref{as:uni} only constrains on the whole trajectory.

We assume $\tilde{C}_{rs}(t)$ is known for all $rs \in K_q, t \in T$. 
The flow-scale-free NDL can be reformulated in Equation \ref{eq:origin2}.
\begin{eqnarray}
\label{eq:origin2}
D_{rs}^{\mathcal{F}}(t)
&=& C_{rs}(t) - \min_{k \in K_{rs}} C_{rs}^{k}(t)\nonumber \\
&=& \max_{k \in K_{rs}} \left(C_{rs}(t) - C_{rs}^{k}(t) \right)
\end{eqnarray}
where
\begin{eqnarray}
C_{rs}(t) = \sum_{k \in K_{rs}} p_{rs}^{k}(t) C_{rs}^{k}(t)
\end{eqnarray}

The OD travel time $C_{rs}(t)$ is the average time to complete a trip for all path flow $Q_{rs}(t)$ on OD pair $r,s$ and departing at time $t$. Base on Assumption~\ref{as:nop}, we use zone-to-zone RV data to estimate $C_{rs}(t)$, as presented in Proposition~\ref{prop:odtt}.

\begin{proposition}
	\label{prop:odtt}
	Based on Assumption~\ref{as:nop},  we have
	\begin{eqnarray}
	\tilde{C}_{rs} (t) \xrightarrow{P} C_{rs}(t)
	\end{eqnarray}
	when $|\tilde{N}_{rs}(t)| \to \infty$.
\end{proposition}

\begin{proof}
	We compute the expectation of $\kappa_{j}$ in each subset $\tilde{N}_{rs}(t)$ based on Assumption~\ref{as:nop},
	\begin{eqnarray}
	\mathbb{E}_{j}\left[\kappa_{j}\right] = \sum_{k \in K_{rs}}p_{rs}^{k}(t) C_{rs}^{k}(t), \forall j \in \tilde{N}_{rs}(t)
	\end{eqnarray}
	By Law of Large Numbers (LLN), we have
	\begin{eqnarray}
	\tilde{C}_{rs} (t) = \frac{1}{|\tilde{N}_{rs}(t)|} \sum_{j \in \tilde{N}_{rs}(t)}  \kappa_{j} \xrightarrow{P} \sum_{k \in K_{rs}} p_{rs}^{k}(t) C_{rs}^{k}(t), \forall j \in \tilde{N}_{rs}(t) = C_{rs}(t)
	\end{eqnarray}
	when $|\tilde{N}_{rs}(t)| \to \infty$.
\end{proof}


Now we are ready to present a new estimator for NDL using zone-to-zone travel time data,
\begin{eqnarray}
\label{eq:agg}
D_{rs}^{\mathcal{A}}(t) = \max_{\tau \in \mathbb{T}}\left(\tilde{C}_{rs}(t) - \tilde{C}_{r \tau s}(t)\right)
\end{eqnarray}
where $\tau$ is a set of connecting (intermediate) points that connect path flows from one OD pair to another, and $\mathbb{T}$ is the set of all possible $\tau$. $\tilde{C}_{r \tau s}(t)$ may be virtual in the sense that there may not exist such an RV which complete the trip following $r$, $\tau$ to $s$ (denoted by $\tilde{C}_{r \tau s}(t)$), but the travel time of this virtual trip can be estimated through the ``relay'' of multiple trips that are actually completed by RVs (in this case, a trip from $r$ to $\tau$, and another trip from $\tau$ to $s$). We define the average travel time among all paths from $r$ to
$s$ through $\tau$ as $\tilde{C}_{r\tau s}(t)$ in Equation~\ref{eq:alter}.

\begin{eqnarray}
\label{eq:alter}
\tilde{C}_{r\tau s}(t) =  \tilde{C}_{r\tau[0]}(t) + \sum_{u = 1}^{|\tau|-1} \tilde{C}_{\tau[u-1]\tau[u]}\left(t^{u}\right)  + \tilde{C}_{\tau[|\tau|]s}\left(t^{|\tau|}\right)
\end{eqnarray}
where $\tau[0]$ is the first element (i.e., connecting or intermediate point) in $\tau$, $\tau[1]$ is the second, etc. Note the travel time of alternative path $C_{r \tau s}(t)$ is defined in the same manner with Equation~\ref{eq:alter}. 
To simplify the computation, we only consider the alternative paths with one connecting point.  When $\tau$ only has one element, the scalar $\tau = \tau[0]$ is used to simplify the notation. Then the travel time of the alternative path $\tilde{C}_{r\tau s}(t)$ with one connecting point is presented in Equation~\ref{eq:alter2}.

\begin{eqnarray}
\label{eq:alter2}
\tilde{C}_{r\tau s}(t) =  \tilde{C}_{r\tau} \left(t\right) + \tilde{C}_{\tau s}\left(t + \tilde{C}_{r\tau}\left(t\right)  \right)
\end{eqnarray}

\begin{example}[Alternative path]
	In this example, we illustrate the idea of alternative path and the definition of $C_{r\tau s}(t)$. The layout of the map and notations are the same as Example~\ref{ex:1}. Figure~\ref{fig:e2} presents the trajectories of three RV trips, and we assume $t = t_{0}^0 = t_{1}^0$ and $t_{2}^0 = t_{1}^2$.
	
	\begin{figure}[h]
		\centering
		\includegraphics[scale = 0.7]{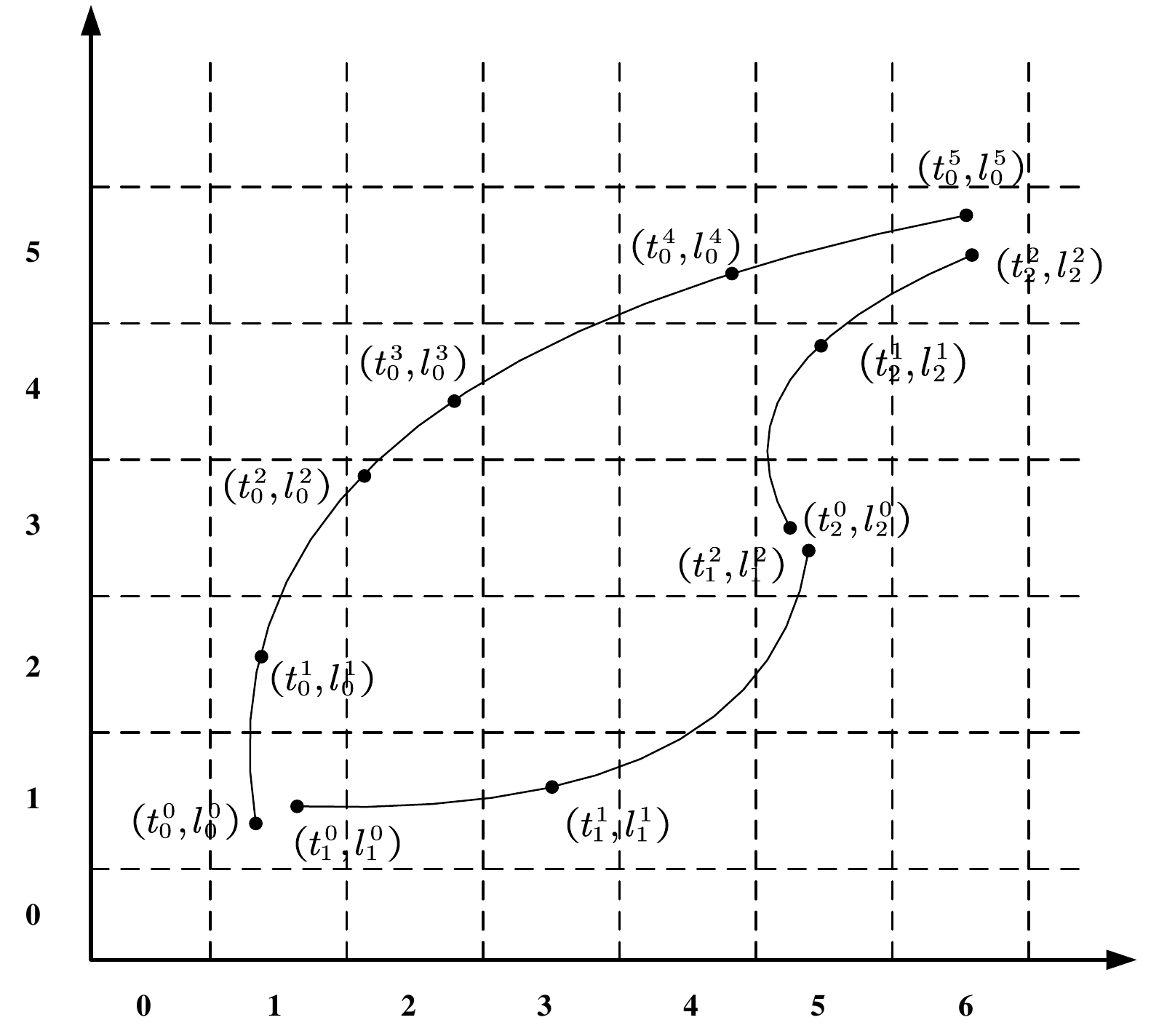}
		\caption{Example of three RV trajectories}
		\label{fig:e2}
	\end{figure}
	
	One can clearly see $t_{0}^5 - t_{0}^0$ is a sample of $C_{(1,1) (6,5)}(t)$. The trajectories $1$ and $2$ form an alternative path from origin $(1,1)$ to destination $(6,5)$, so $t_{2}^2 - t_{0}^0$ is a sample of $C_{(1,1) (5,3) (6,5) }(t)$. We note here $\tau = (5,3)$ (or $\tau = [(5,3)]$ if we do not use a simplified form), meaning zone $(5,3)$ is a connecting point.
\end{example}

We further prove the asymptotic property of the estimated NDL by zone-to-zone travel time in Proposition~\ref{prop:agg2}.

\begin{proposition}
	\label{prop:agg2}
	When $|\tilde{N}_{rs}(t)| \to \infty$ and $\mathbb{T}$ enumerates all the possible selections of $\tau$, we have
	\begin{eqnarray}
	D_{rs}^{\mathcal{A}}(t)\xrightarrow{P} D_{rs}^{\mathcal{F}}(t)
	\end{eqnarray}
\end{proposition}
\begin{proof}
	By Proposition~\ref{prop:odtt}, we have
	\begin{eqnarray}
	\tilde{C}_{rs} (t) &\xrightarrow{P}& C_{rs}(t)\\
	\tilde{C}_{r \tau s} (t) &\xrightarrow{P}& C_{r \tau s}(t)
	\end{eqnarray}
	We also have
	\begin{eqnarray}
	\min_{\tau \in \mathbb{T}} C_{r \tau s}(t) &=&\min_{k \in K_{rs}} C_{rs}^k(t)
	\end{eqnarray}
	by choosing $\tau$ as the list of intermediate points of path $k$.
	Again by continuous mapping theorem, the proposition holds.
\end{proof}

\section{Traffic management through user optimal routing with limited control ability}
\label{sec:routing}
In this section, we present an NDL-based traffic management framework. The framework aims at achieving user optimal routing through minimizing NDLs in the network. 
The mathematical formulation for the user optimal routing problem is presented in Equation~\ref{eq:target}.

\begin{equation}
\label{eq:target}
\begin{array}{rrclclc}
\vspace{5pt}
\displaystyle \min_{ \{F_{rs}^{k}(t)\}_{rskt}} & \multicolumn{4}{l}{\displaystyle  \int_{t \in T}\sum_{rs \in K_q}Q_{rs}(t) C_{rs}(t)dt - \int_{t \in T}\sum_{rs \in K_q}Q_{rs}(t) \Pi_{rs}(t) dt} \\
\textrm{s.t.} & p_{rs}^{k}(t) Q_{rs}(t) & = & F_{rs}^{k}(t) & \forall rs \in K_q, k \in K_{rs}, t \in T\\
& p_{rs}^{k}(t) &=& \Psi_{rs}^{k}(t, C_t)&\forall rs \in K_q, k \in K_{rs}, t \in T\\
& C_{rs}(t) &=& \Omega_{rs}^k(t, F_t) & \forall rs \in K_q,  k \in K_{rs}, t \in T\\
& F_{rs}^k(t) &\geq& 0 & \forall rs \in K_q, k \in K_{rs}, t \in T
\end{array}
\end{equation}
where $\Psi_{rs}^{k}(t, \cdot)$ is the route choice portion function for path $k$ in OD $rs$ departing at time $t$, and $\Omega_{rs}^k(t, \cdot)$ denotes the path cost function for path $k$ in OD $rs$ departing at time $t$. $C_t$ and $F_t$ represent all the cost and path flow information before time $t$, respectively.

Instead of system optimal routing, the proposed traffic management framework aims at the user optimal routing due to the following reasons: 1) {\em Travelers' compliance issue}: Under the system optimal routing, some travelers who are routed to a longer route may not follow the guidance; by contrast, travelers tend to follow the guidance under user optimal routing since the provided route is always the shortest; 2) {\em Fairness issue}: similar to 1), the system optimal routing can be unfair to some travelers \citep{schulz2006efficiency}; 3) {\em Small gap}: the gap between user optimal and system optimal is small in terms of total travel time \citep{roughgarden2002selfish}; 4) {\em Data constraints:} current RV dataset does not support an effective system optimal routing method.

Now we analyze the objective function in Equation \ref{eq:target} by Proposition~\ref{prop:deco}.

\begin{proposition}[Total travel time decomposition]
	\label{prop:deco}
	The total travel time for each OD pair departing at time $t$ can be decomposed into two parts: the merit function-based NDL and the minimum achievable total travel time, as presented in Equation~\ref{eq:decop}.
	\begin{eqnarray}
	\label{eq:decop}
	Q_{rs}(t) C_{rs}(t)   = D_{rs}^{\mathcal{M}}(t)   + Q_{rs}(t) \Pi_{rs}(t)
	\end{eqnarray}
	where $Q_{rs}(t) \Pi_{rs}(t)$ is denoted as the minimum achievable total travel time for OD pair $rs$ departing at time $t$.
\end{proposition}
\begin{proof}
	\begin{eqnarray}
	Q_{rs}(t) C_{rs}(t)  &=&
	\left[Q_{rs}(t) C_{rs}(t) - Q_{rs}(t) \Pi_{rs}(t)  \right] +  Q_{rs}(t) \Pi_{rs}(t)  \\
	&=&D_{rs}^{\mathcal{M}}(t)   + Q_{rs}(t) \Pi_{rs}(t)
	\end{eqnarray}
\end{proof}

We substitute Equation~\ref{eq:decop} to Equation~\ref{eq:target}, then the original optimization problem is simplified to Equation~\ref{eq:dis_target} (route choice and travel time constraints are omitted).

\begin{equation}
\label{eq:dis_target}
\begin{array}{rrclclc}
\vspace{5pt}
\displaystyle \min_{\{F_{rs}^{k}(t)\}_{rskt}} & \multicolumn{4}{l}{\displaystyle \int_{t \in T}\sum_{rs \in K_q} D_{rs}^{ \mathcal{M}}(t) } \\
\textrm{s.t.} & \sum_{k \in K_{rs}} F_{rs}^{k}(t)& = & Q_{rs}(t) & \forall t \in T, rs \in K_q\\
& F_{rs}^k(t) &\geq& 0 & \forall t \in T, rs \in K_q, k \in K_{rs}
\end{array}
\end{equation}

Equation \ref{eq:dis_target} claims that user optimal routing problem is equivalent to NDL minimizing problem, hence it can be solved by adaptive ({\em i.e} reactive) user optimal routing using the instantaneous travel information \citep{kuwahara2001dynamic}. However, controlling all the vehicles may not be feasible in real-world networks. Instead, many studies have shown that traffic delay can be considerably reduced by controlling a relatively small portion of vehicles \citep{pi2017stochastic,zhang2018mitigating, sharon2018traffic}. Our approach focuses on solving the user optimal routing with limited control ability using NDL measures of sampled vehicles.

We note formulation~\ref{eq:target} and \ref{eq:dis_target} are conceptual formulations because $\Pi_{rs}(t)$ cannot be obtained in real-world networks. Both formulations are used to exploit the connection between the user optimal routing and NDL minimization problem, and the connection further yields the real-time management framework presented in the following context.







We now propose a real-time traffic management framework to achieve user optimal routing with limited resources. In general, the frameworks is similar to a standard  user optimal routing method that routes travelers to the optimal paths. The major contribution of the proposed framework is that NDL can be estimated in real time and used to prioritize routing vehicles/trips. NDLs provide a routing method that identifies the path flow that is the most distant from user optimum (i.e. user optimal paths), and controlling vehicles/trips from those OD pairs is likely to reduce  traffic congestion the most effectively based on Equation~\ref{eq:dis_target}.

Generally, the algorithm lean to control the flow along an OD pair with highest NDL, such that the limited control resources can be fully utilized. Provided that a routing platform is able to control up to $\alpha$ portion of the total demand,
a traffic management method is presented in Algorithm~\ref{alg:control}.

\begin{algorithm}[H]
	\SetKwInOut{Input}{Input}
	\SetKwInOut{Output}{Output}
	\underline{Routing} $\left(D_{rs}^{\mathcal{F}}(t), F_{rs}^{k}(t), \alpha, \texttt{lag} \right)$\;
	\Input{NDL $D_{rs}^{\mathcal{F}}$, path flow $F_{rs}^{k}(t)$, control ratio $\alpha$, delay of NDL information $\texttt{lag}$}
	\Output{Adjusted path flow $\tilde{F}_{rs}^{k}(t)$}
	Initialize $\tilde{F}_{rs}^{k}(t) = F_{rs}^{k}(t), \forall rs \in K_q, k \in K_{rs}$\;
	Initialize list $L$ as an empty list\;
	Compute controllable total flow $R = \alpha \sum_{rs} \sum_{k} F_{rs}^{k}(t)$\;
	Sort OD pairs by $D_{rs}^{\mathcal{F}}(t-\texttt{lag})$ in descending order and store the list in $L$\;
	\While{$R > 0$ and $|L| > 0$}{
		Pop the OD pair $rs$ with highest  $D_{rs}^{\mathcal{F}}(t-\texttt{lag})$ in $L$\;
		Search for the shortest path $k$ based on $v_i, i \in N_{rs}(t-\texttt{lag})$ for trajectory-level information or $C_{r \tau s}(t-\texttt{lag})$ for zone-to-zone travel time\;
		Route $F_{rs}^{k'}(t), k' \neq k$ to path $k$, $\tilde{F}_{rs}^{k}(t) = \tilde{F}_{rs}^{k}(t) + \sum_{k' \neq k} F_{rs}^{k'}(t)$, $\tilde{F}_{rs}^{k'}(t) = 0$\;
		$R = R - \sum_{k' \neq k} F_{rs}^{k'}(t)$\;
	}
	\caption{NDL-based user optimal routing algorithm}
	\label{alg:control}
\end{algorithm}

In the algorithm, $\alpha$ denotes the maximum ratio of total demand that can be controlled, path flow $F_{rs}^{k}(t)$ is the original path flow without any control, $\tilde{F}_{rs}^{k}(t)$ is the path flow after control (i.e. routing).
Note the vehicles controlled in each time $t$ can be different.
The NDL $D_{rs}^{\mathcal{F}}(t)$ is estimated by RV data based on Equation~\ref{eq:dis2} or Equation~\ref{eq:agg}, and its estimation is usually delayed, meaning that  $D_{rs}^{\mathcal{F}}(t)$ is not immediately available at time $t$. It may be available after $\texttt{lag}$ amount of time intervals. To address this issue, any time series prediction methods can be adopted to obtain a forecast of $D_{rs}^{\mathcal{F}}(t)$ . In the experiments, we will show that simply using $D_{rs}^{\mathcal{F}}(t - \texttt{lag})$ as a prediction for $D_{rs}^{\mathcal{F}}(t)$ can achieve reasonably effective congestion reduction. The term $\texttt{lag}$ refers to the time delay of NDL estimation, and $D_{rs}^{\mathcal{F}}(t - \texttt{lag})$ is the latest NDL available for OD pair $rs$ at time $t$.

\section{Case study I: metropolitan Chengdu area with DiDi chuxing data}
\label{sec:didi}

In this section, we conduct a case study with RV trajectory data in Chengdu, China. Chengdu is a sub-provincial city which serves as the capital of Sichuan Province. As of 2014, its urban population is around $10$ million. Chengdu also serves as a tourism city with multiple large shopping districts, food plazas and cultural attractions \citep{wikichengdu}. Therefore, both commuting traffic and non-recurrent traffic by tourists are critical when analyzing traffic conditions in Chengdu.

The dataset we use includes the trajectory of all RVs operated by DiDi Chuxing from November 1st, 2016 to November 30th, 2016.
Each trajectory record is represented by a sequence of temporal and spatial stamps for a trip order containing a pick-up and a drop-off.
Since the network topology and traffic flow data in Chengdu are not available, the effectiveness of the NDL-based traffic management method cannot be tested in this case study. Hence the objective for first case study is to demonstrate the spatio-temporal patterns of NDL. We will conduct a second case study in Section 7 to exam the traffic management algorithm.

The trajectory data contains duplicate records and outliers, which need to be cleaned before processing. The duplicates are identified by the unique order ID, and the outliers are identified by unrealistic trajectory and travel time. The metropolitan area of Chengdu is segmented into 545,305 cells in the grid, and each cell is a $300$m $\times 300$m square. The origin and destination of each RV trip are matched to the cells by  $\xi(\cdot)$ function defined in Equation~\ref{eq:match}. For example, Figure~\ref{fig:didi_tra} demonstrates how the trajectories look like for one OD pair. Part of the trajectory may be off the road because of measuring errors, but it is clearly that two RVs take the same route and a third RV takes a different route. 

\begin{figure}[h]
	\centering
	\includegraphics[scale = 0.25]{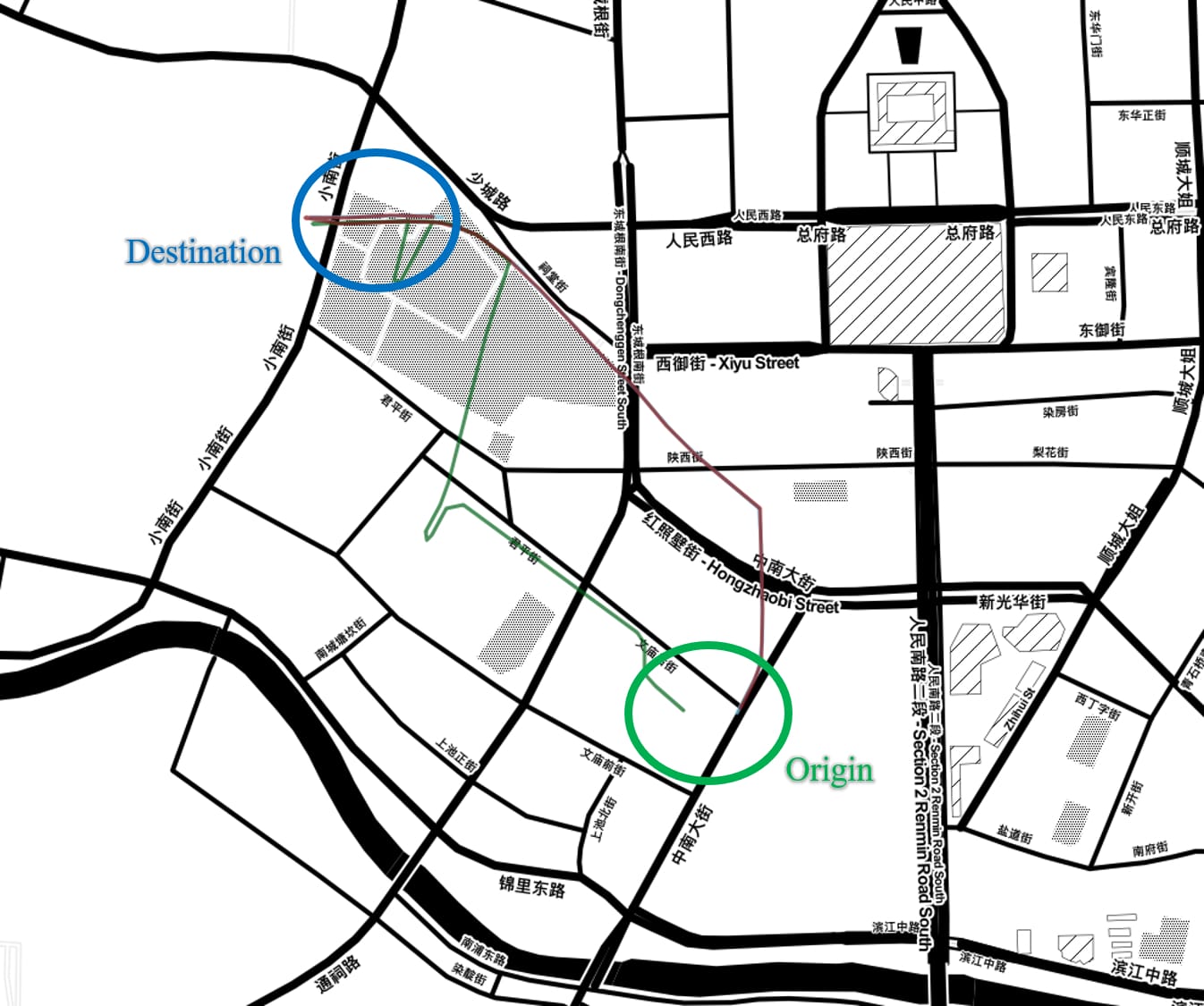}
	\caption{Three trajectories with the same OD in 11:00-12:00, Nov 1st, 2016}
	\label{fig:didi_tra}
\end{figure}
\subsection{NDL aggregated over all OD pairs}
\label{sec:didi_NDL_Ave}
With the cleaned and processed RV data, we compute the NDL measure using Equation~\ref{eq:dis2} for each hour. We first visualize the average NDL aggregated over all OD pairs on each day. In this sub-section, we will use the term ``average NDL'', which is defined as $\frac{1}{|K_q|} \sum_{rs \in K_q} D_{rs}(t)$. Note the ``average NDL'' reflects the trend of NDL in different time intervals, while it does not reflect the average difference between current travel time and the shortest travel time. The latter one is denoted by $\frac{1}{\sum_{rs \in K_q} Q_{rs}(t)} \sum_{rs \in K_q} Q_{rs}(t) D_{rs}(t)$ instead, and it can not be obtained as the demand $Q_{rs}(t)$ is unknown in this study.

\subsubsection{Weekdays v.s. Weekends}


We examine the weekly patterns of NDL for weekdays and weekends. We plot time-of-day NDL for each day (in transparent colors), along with the daily average (in solid colors), in Figure \ref{fig:didi_aveNDL}. The average NDL patterns on weekdays and weekends are quite different. There are three major spikes on weekdays, two corresponding to the morning peaks and one corresponds to the afternoon peak. There is only
one spike on weekends, and the average NDL steadily increases from 4:00am to 14:00pm.

As an approximation to $D_{rs}^{\mathcal{F}}(t)$, the estimated NDL is independent of travel demand. However, the average NDL and total demand follow similar time-of-day patterns in Chengdu, which is not always the case as we will show later in another example. The positive correlation between NDLs and travel demand indicates that high demand level can induce high NDLs in this case. The spikes of NDL can also be explained by the nature of travel demand for the Chengdu network. On weekdays, the two spikes in morning peaks are probably due to the various working start time of travel demand, and the spike in afternoon peaks is due to the evening commute. On weekends, the only spike is largely attributed to traffic demand for weekend entertainments.


\begin{figure}[h!]
	\centering
	\begin{subfigure}[b]{0.475\textwidth}
		\includegraphics[width=\textwidth]{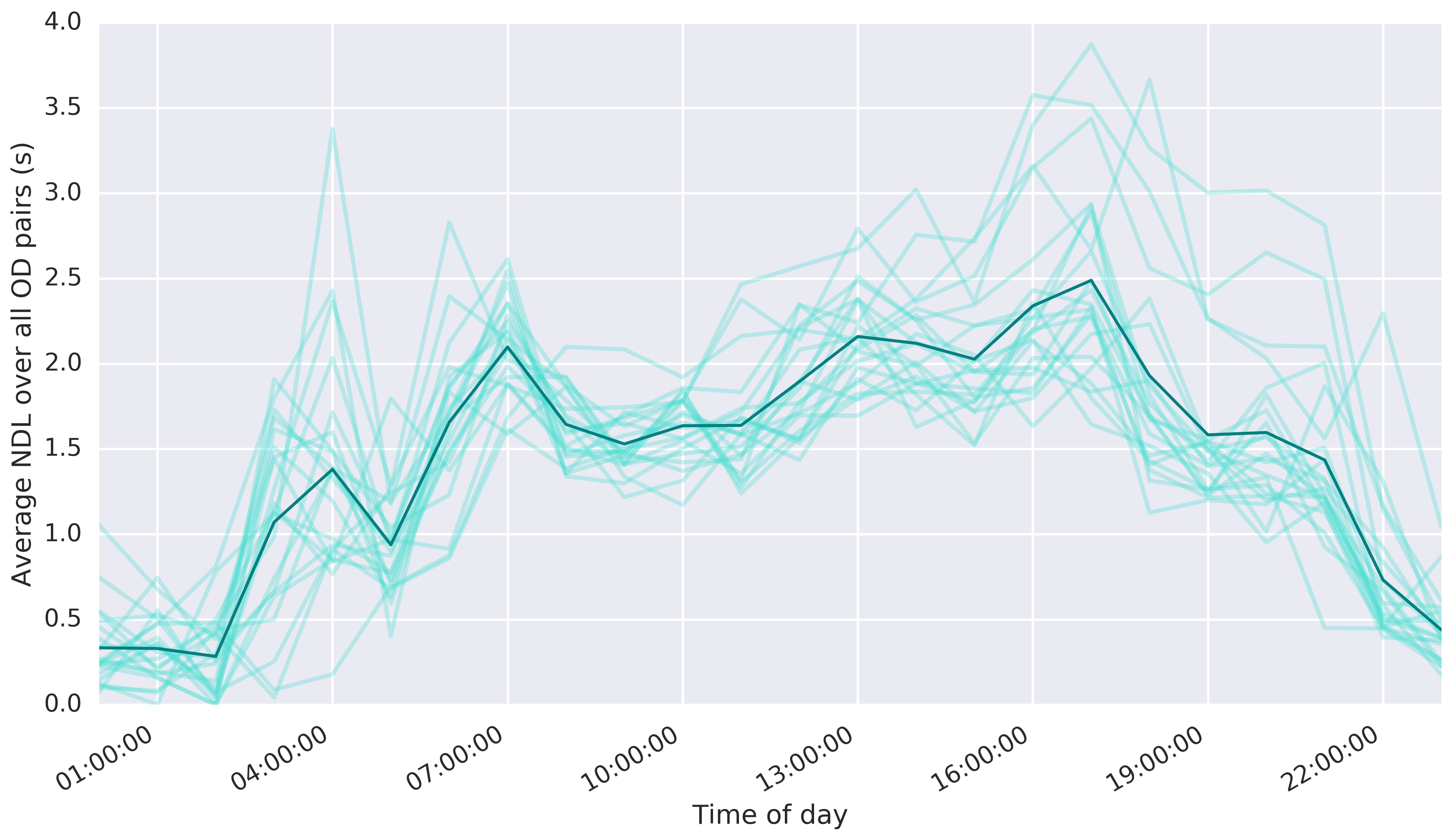}
		\caption{\footnotesize{Weekdays}}
	\end{subfigure}
	\begin{subfigure}[b]{0.475\textwidth}
		\includegraphics[width=\textwidth]{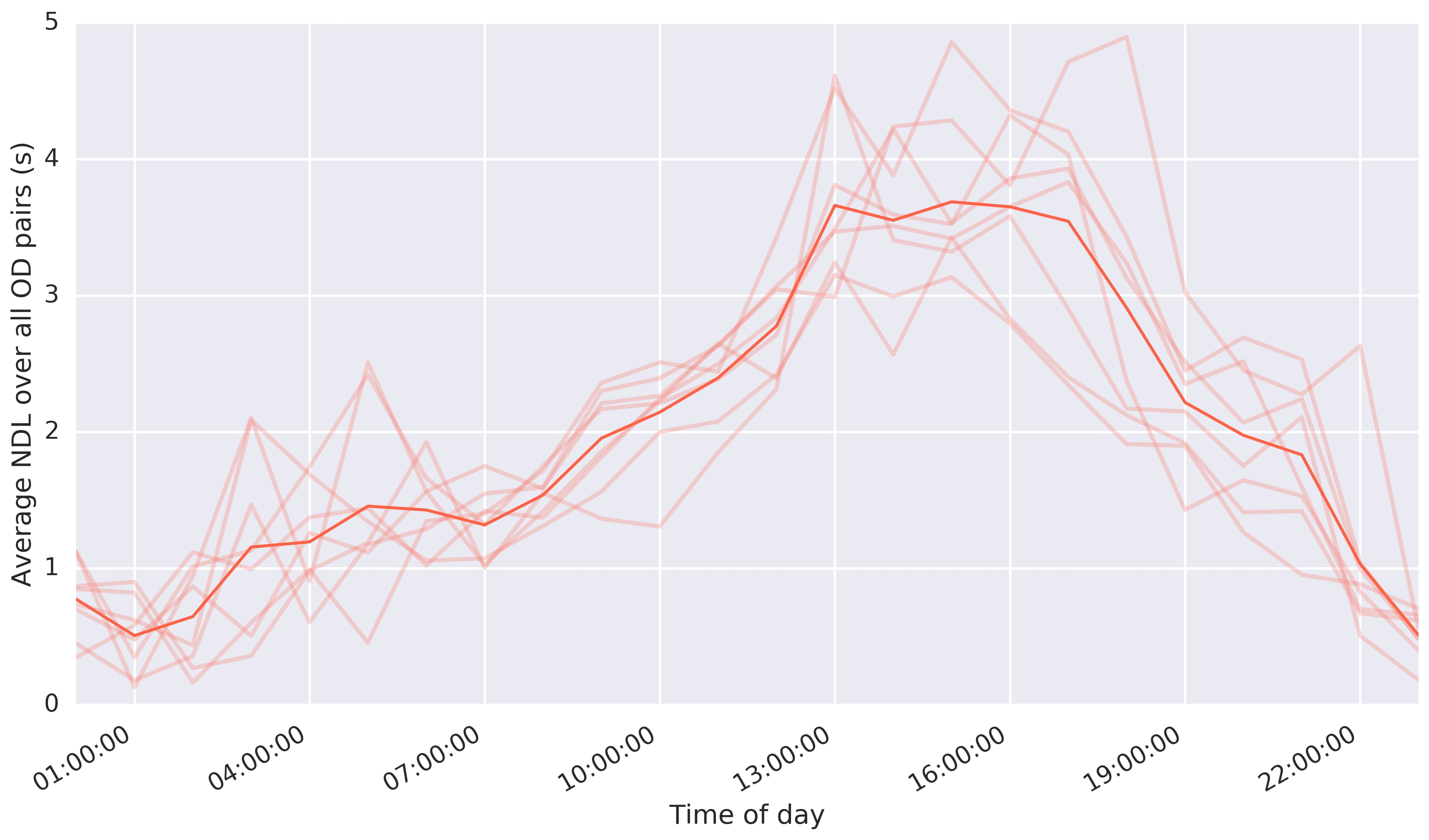}
		\caption{\footnotesize{Weekends}}
	\end{subfigure}
	\caption{\footnotesize{Average NDL by time of day, on weekdays and weekends (solid lines are the average of average NDL taken over all weekdays and weekends, respectively)}}
	\label{fig:didi_aveNDL}
\end{figure}

\subsubsection{Day of week and time of day effects on NDL measure}
\label{sec:didi_time}
We compute the average NDL by hour, averaged over all days in each day of week, as well as the percentage change in average NDL by hour where the base is set as the average of NDL taken
over all days of a week. The results are presented in Figure~\ref{fig:didi_compare}.

\begin{figure}[h!]
	\centering
	\begin{subfigure}[b]{0.7\textwidth}
		\includegraphics[width=\textwidth]{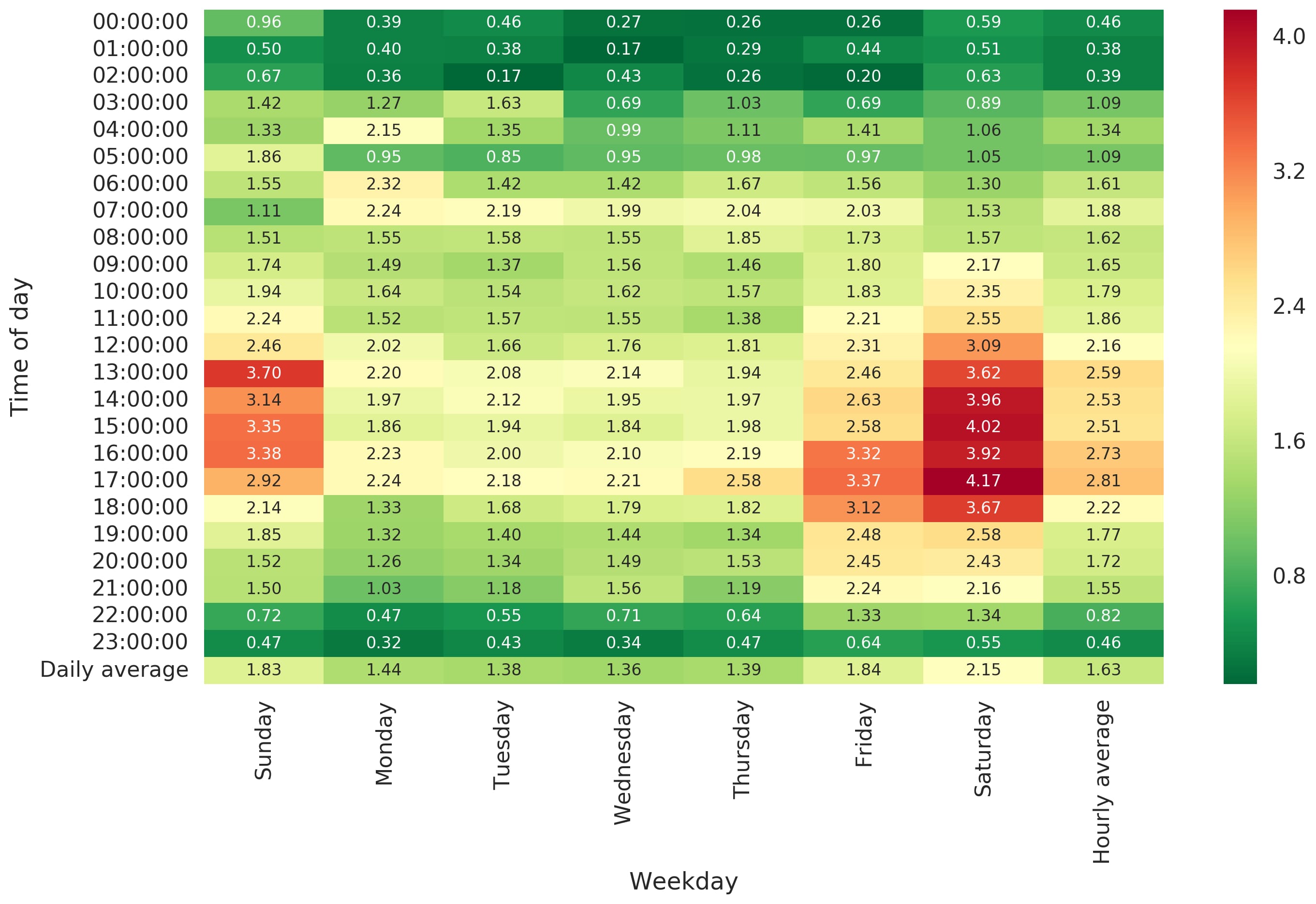}
		\caption{\footnotesize{Average NDL}}
		\label{fig:didi_compare1}
	\end{subfigure}
	\begin{subfigure}[b]{0.7\textwidth}
		\includegraphics[width=\textwidth]{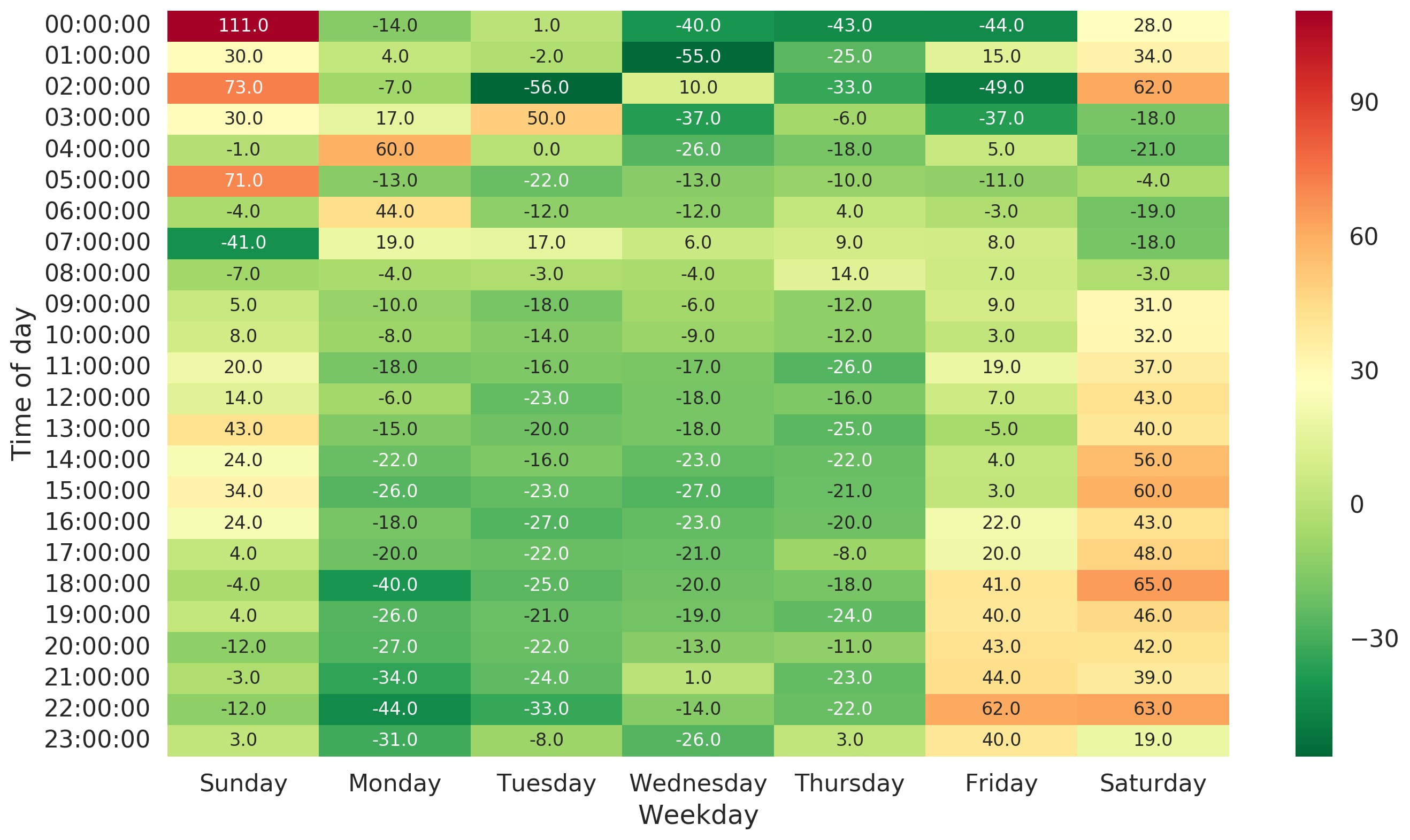}
		\caption{\footnotesize{Percentage change (\%)}}
		\label{fig:didi_compare2}
	\end{subfigure}
	\caption{\footnotesize{Average NDL and its percentage change by hour by day of week, comparing to the daily average of NDL taken over
			all days of week }}
	\label{fig:didi_compare}
\end{figure}

Travelers on the weekdays are largely attributed to recurrent commuters, so the traffic condition is considered to be more stable than the weekends. As a result, NDL is generally lower on weekdays. On the contrary, weekend activities are less likely to follow a fixed pattern for travelers, hence the NDL becomes much higher. In particular, NDL during midnight of weekends is significantly higher than that on weekdays, possibly as a result of weekend midnight activities.

\subsubsection{Spatial pattern of average NDL}
\label{sec:spa_spar}
We compute the NDL between a random subset of OD pairs averaged over each day, and present the results in Figure~\ref{fig:didi_odpair}. To improve the visualization effects, we visualize the NDL lower than $50$s all in the same white color.

\begin{figure}[h]
	\centering
	\includegraphics[scale = 0.1]{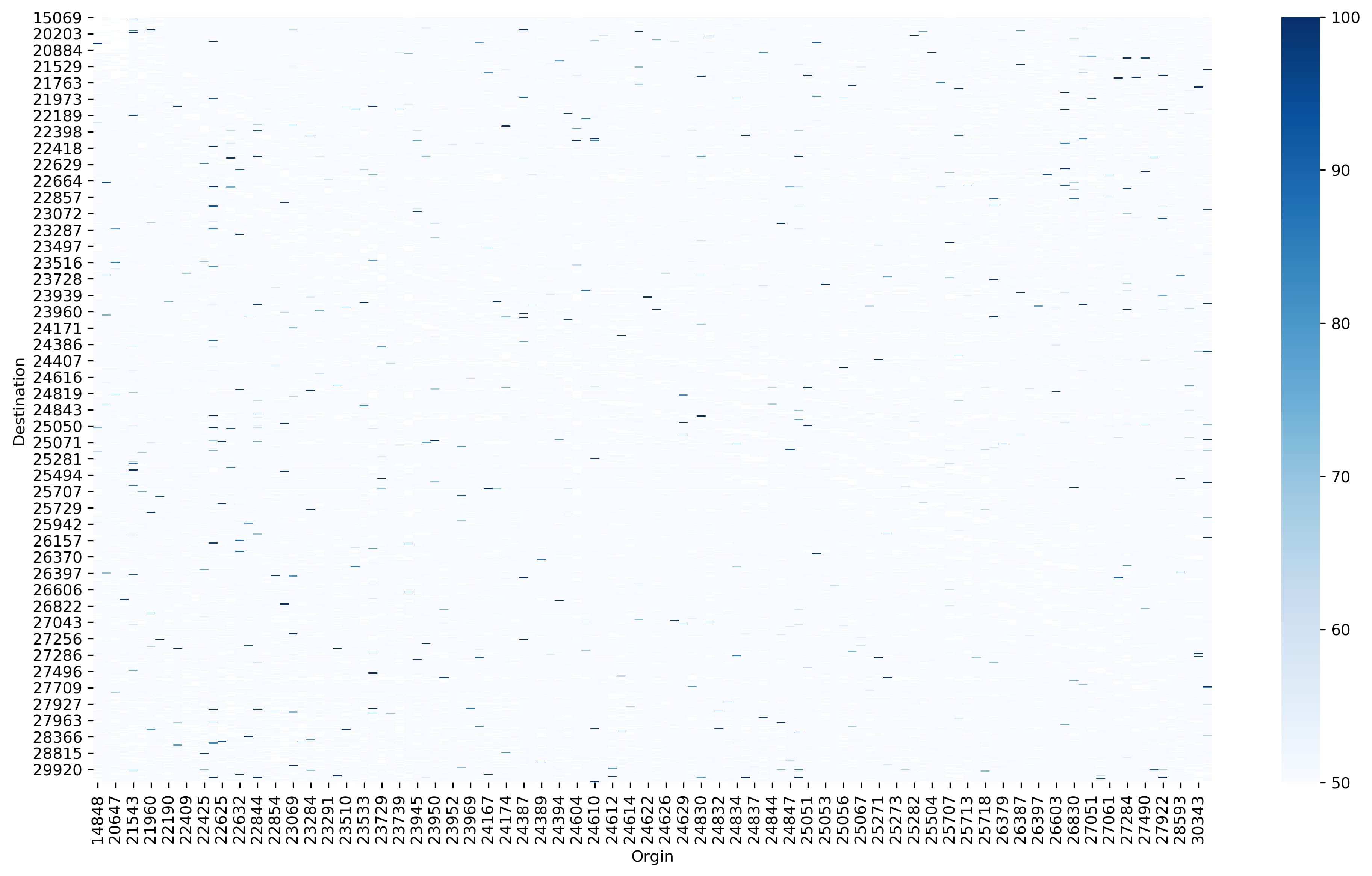}
	\caption{Average NDL between a random subset of OD pairs (unit: seconds)}
	\label{fig:didi_odpair}
\end{figure}

As can be seen from Figure~\ref{fig:didi_odpair}, NDL is sparsely distributed among all OD pairs. The NDL can be as high as $100$s, while most of NDLs are nearly zero.
The sparsity and high variability of NDL between OD pairs are attributed to the demand pattern that are spatially imbalanced and inefficient/inaccurate information available to demand for certain OD pair. Those OD pairs with high NDL can be considered as ``bottlenecks'' of the urban networks, and hence controlling flow along these crucial OD pairs can considerably help traffic management.

\subsection{NDL for individual OD pairs}
Now we examine the spatio-temporal NDL of each OD pair over the course of the study time horizon.

\subsubsection{Time of day effects on NDL measure for single OD pair}
\label{sec:todNDL}
We draw a figure with ($n \times m$) pixels, $n$ is the number of days and $m$ is the number of time intervals on each day.
We set $y$ axis to be the dates in Nov. 2016, and $x$ axis to be the time of day from 00 : 00 to 23 : 59. Each pixel
is color coded to indicate the value of NDL, and color grey means the NDL is not available. This figure demonstrates the daily time-of-day NDL change over the month for each OD pair in high granularity. We randomly selected $12$ OD
pairs and plot them in Figure~\ref{fig:didi_12od}. Generally, NDLs on most of days follow a similar pattern. One can clearly see NDL measures of OD pairs $(22615, 22844)$ and $(22845, 22616)$ are higher than those of OD pairs $(23754, 25514)$ and $(28361, 30362)$. In addition to the NDL spatial sparsity presented in section~\ref{sec:spa_spar}, the occurrence of high NDL for each OD pair is also sparse.

\begin{figure}[h]
	\centering
	\includegraphics[scale = 0.08]{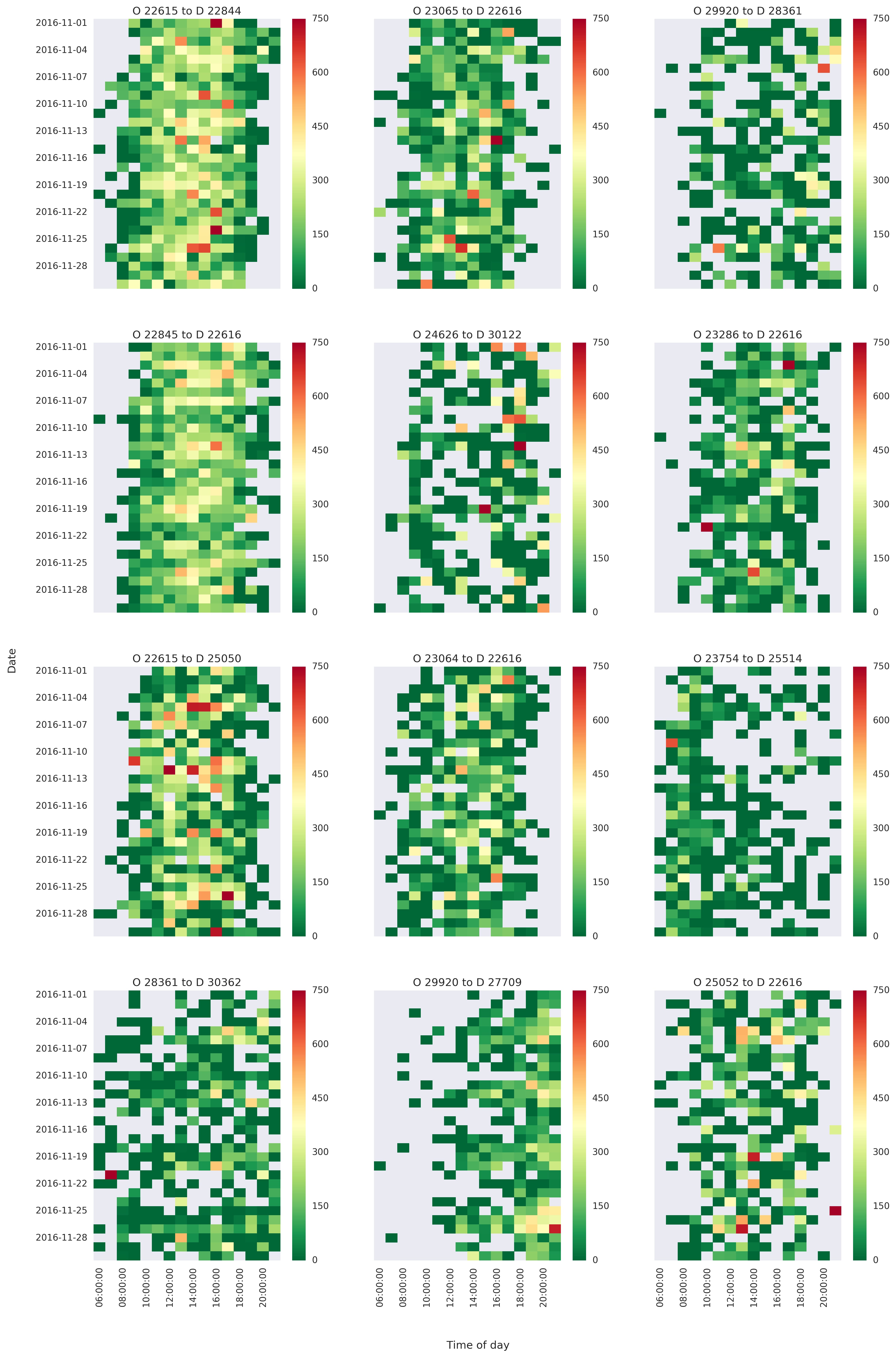}
	\caption{Time-of-day NDL profile for randomly selected $12$ OD pairs}
	\label{fig:didi_12od}
\end{figure}

Another interesting observation is that the NDLs for different ODs are quite random in one specific hour of day, though the general patterns aggregated over all ODs are similar. The randomness of NDL at each particular location and time reflects the nature of dynamic stochastic networks: traffic conditions of each specific time interval and location may be sensitive to demand patterns and non-recurrent factors such as accidents and events. However, NDLs aggregated in the network still exhibit a spatial and time-of-day pattern.

\subsubsection{Origin-base  and destination-based NDL measures}
\label{sec:od}

Before presenting the visualization results, we first overview the layout of Chengdu. Figure~\ref{fig:didi_map} presents an overview of the Chengdu network and some of its points of interests (POIs). 

\begin{figure}[h]
	\centering
	\includegraphics[scale = 0.20]{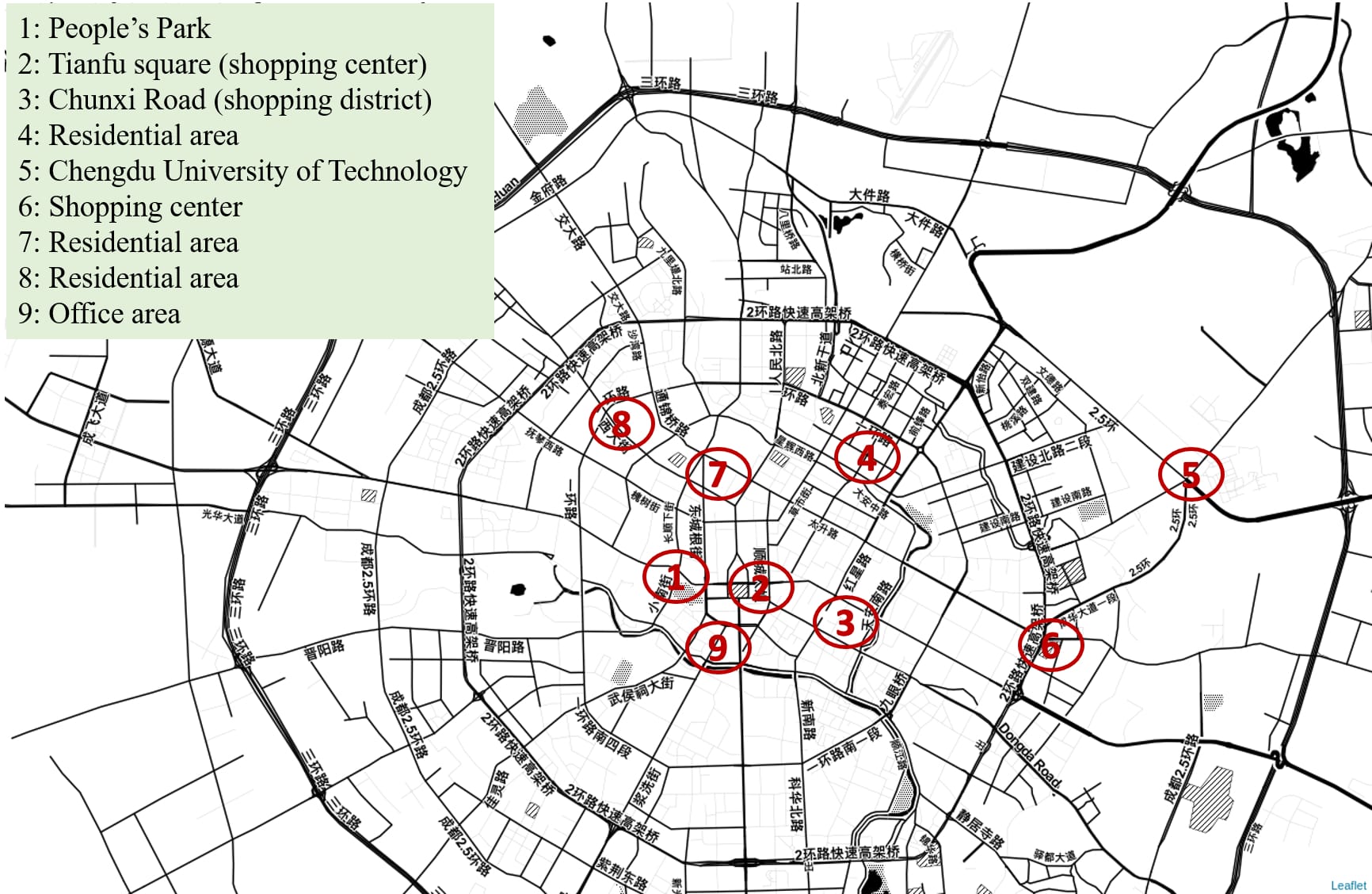}
	\caption{Overview of Chengdu network and its POIs}
	\label{fig:didi_map}
\end{figure}

We compute the origin-based and destination-based NDL by hour with Equation~\ref{eq:odndl1} and \ref{eq:odndl2}, averaged over all weekdays and weekends, and the results are presented in the supplementary materials. Each figure is color coded from green to red, representing the NDL from low to high. We select four interesting figures presented in Figure~\ref{fig:od_results}.

\begin{figure}[h]
	\centering
	\begin{subfigure}[b]{0.485\textwidth}
		\includegraphics[width=\textwidth]{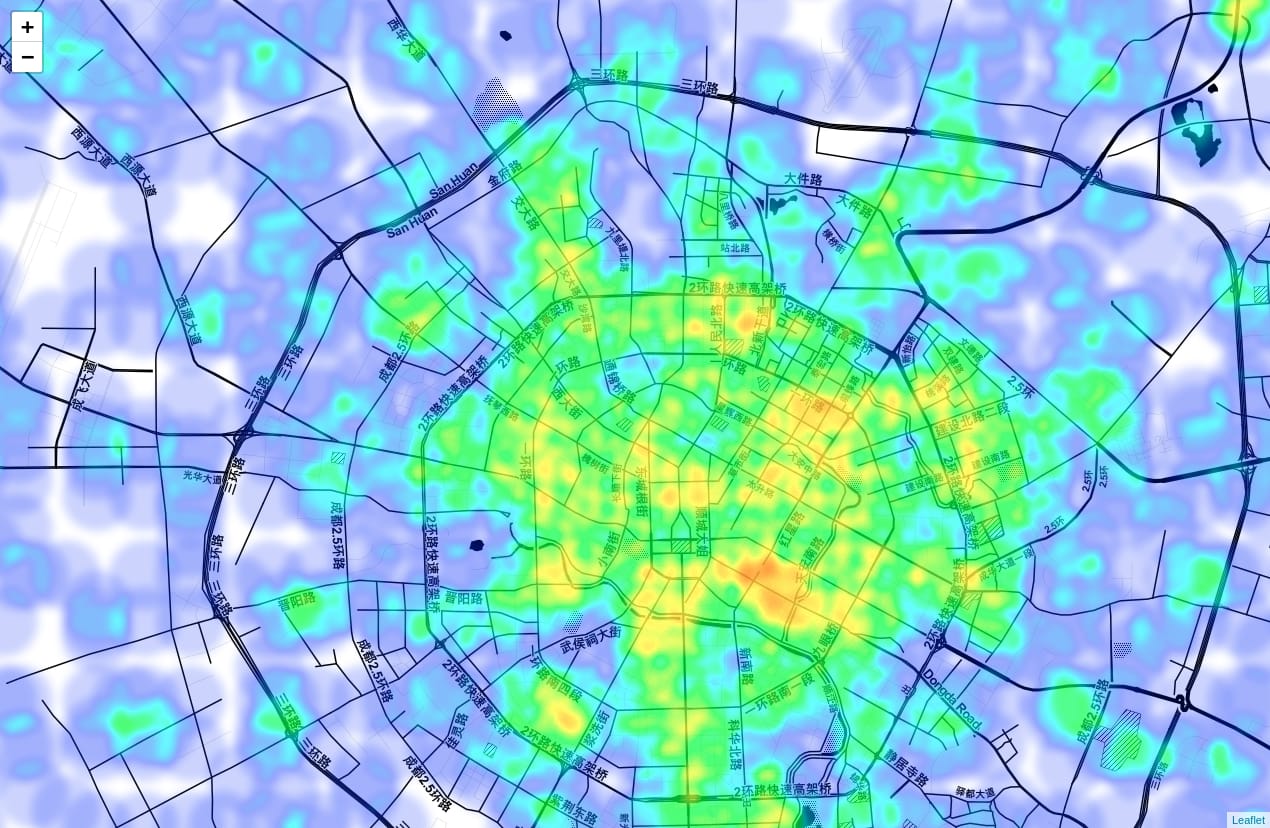}
		\caption{\footnotesize{Destination-based NDL from 8:00 AM to 12:00 AM on weekdays}}
		\label{fig:od_results1}
	\end{subfigure}
	\begin{subfigure}[b]{0.485\textwidth}
		\includegraphics[width=\textwidth]{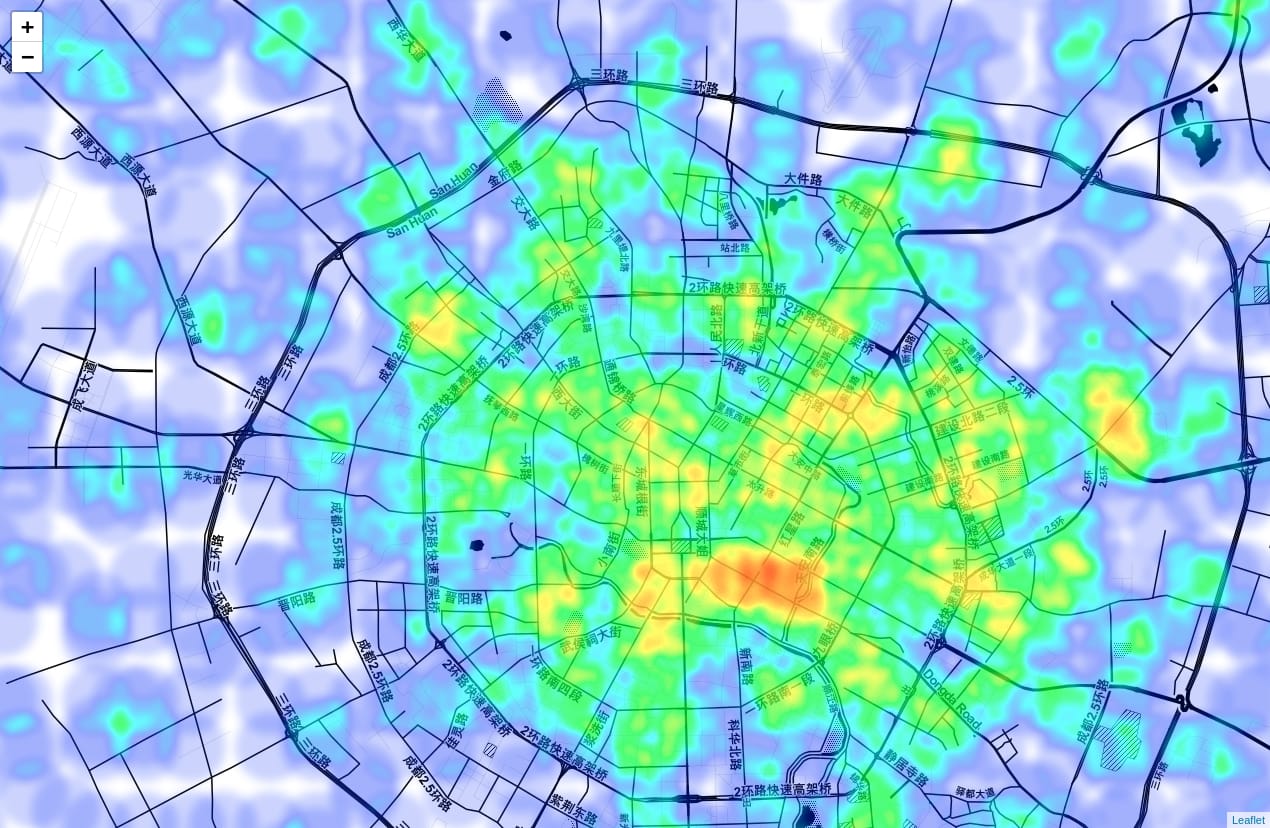}
		\caption{\footnotesize{Origin-based NDL from 16:00 PM to 20:00 PM on weekdays}}
		\label{fig:od_results2}
	\end{subfigure}
	\begin{subfigure}[b]{0.485\textwidth}
		\includegraphics[width=\textwidth]{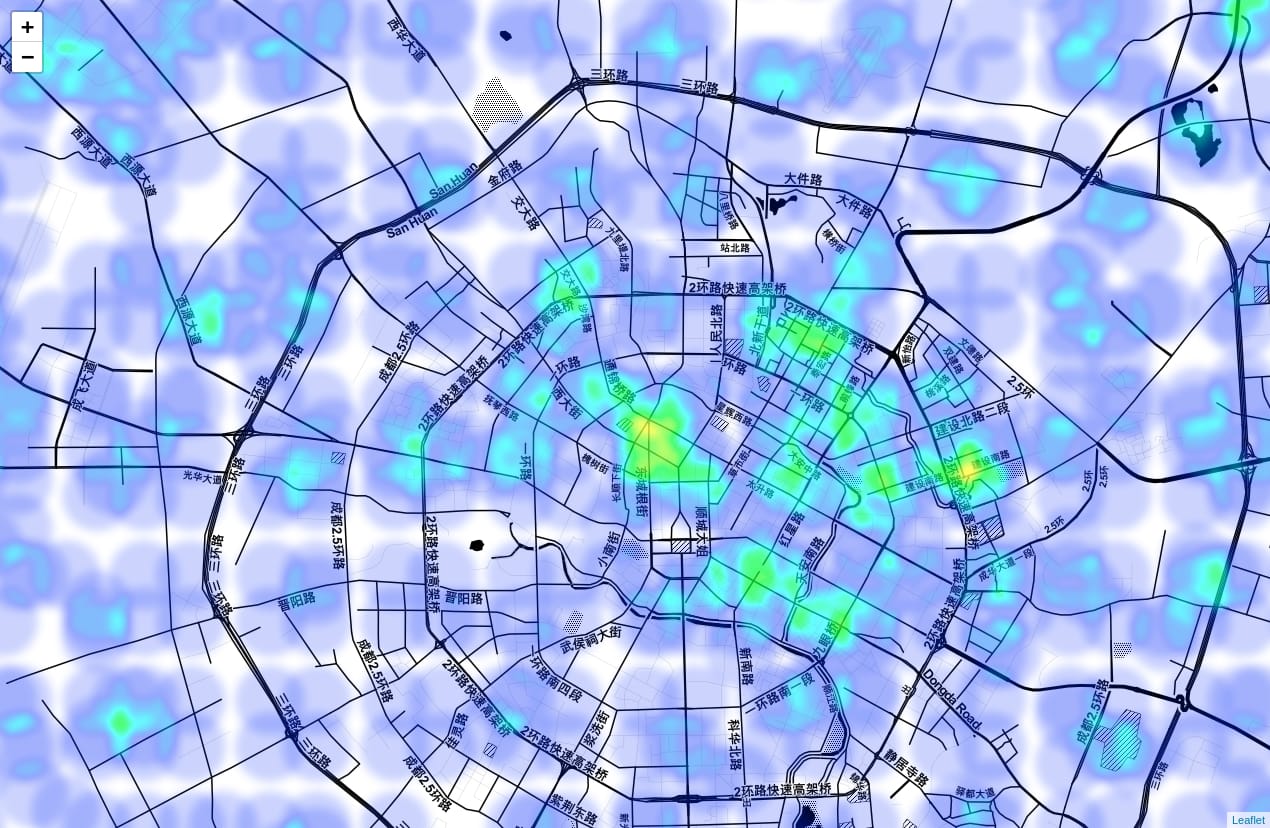}
		\caption{\footnotesize{Destination-based NDL from 0:00 AM to 4:00 AM on weekends}}
	\end{subfigure}
	\begin{subfigure}[b]{0.485\textwidth}
		\includegraphics[width=\textwidth]{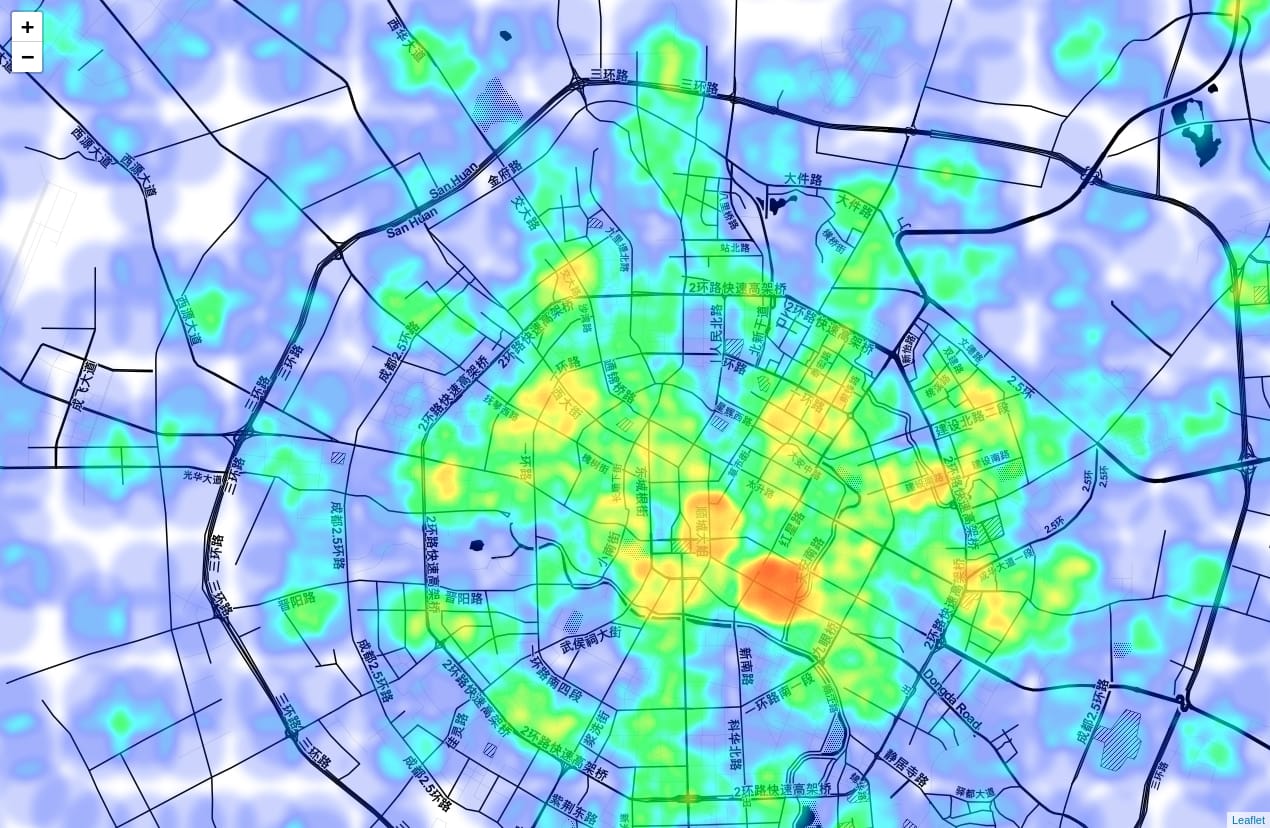}
		\caption{\footnotesize{Destination-based NDL from 16:00 PM to 20:00 PM on weekends}}
	\end{subfigure}
	\caption{\footnotesize{Selected origin and destination based NDL (Red color represents high NDL and green color represents low NDL)}}
	\label{fig:od_results}
\end{figure}

Generally, NDLs for all the sightseeing spots and shopping districts are considerably higher than that for residential areas and office areas, on both weekdays and weekends. One can read from Figure~\ref{fig:od_results1} that most of the inefficient trips are associated with destination Chunxi Road ($3$), which is the largest shopping district in Chengdu. The destination-based NDL in Figure~\ref{fig:od_results1} and the origin-based NDL in Figure~\ref{fig:od_results2} have a similar pattern, since they represent the same group of travelers during the morning peak and afternoon peak. In addition, the NDL for Chunxi Road in the afternoon is higher than that in the morning, because more travelers head for night activities from region $9$ (an office area).

\cleardoublepage

\section{Case study II: Pittsburgh metropolitan area with Uber Movement data}
\label{sec:uber}
In this section, we conduct a case study using zone-to-zone travel time data in Pittsburgh provided by the Uber Movement. Pittsburgh is a city in the Commonwealth of Pennsylvania with an urban population of $2.3$ million \citep{wikipitts}. Pittsburgh is known for its medical, education, manufacture and high-tech industries, more so than a tourist city.

In 2017, Uber released the zone-to-zone travel time data in Pittsburgh on Uber Movement. This dataset is consistent with our proposal of data sharing scheme. The Uber Movement data contains the mean, minimum and maximum of $\kappa_j$ for all TAZ pairs (or census tract pairs) over the years. The Uber Movement data we use here includes the hourly zone-to-zone travel time within the Allegheny County from January 1st, 2016 to June 30th, 2017. In this case study, we omit those results that are similar to Case Study I, and only highlight the differences between Pittsburgh and Chengdu.

Following the steps described in section~\ref{sec:agg}, NDLs can be computed based on Equation~\ref{eq:agg}.

\subsection{NDL aggregated over all OD pairs}
We start with the average NDL aggregated over all OD pairs on each day. We note that the definition of ``average NDL'' is the same as that in section~\ref{sec:didi_NDL_Ave}.
\subsubsection{Weekdays v.s. Weekends}

We plot time-of-day average NDL for each day (in transparent colors), along with the daily average (in solid colors), in Figure \ref{fig:uber_aveNDL}. As can be seen, different from the case in Chengdu, the NDL in Pittsburgh on weekdays is fairly close and stable both from hour to hour and from day to day, except for midnight to early morning. The NDL in Pittsburgh shows a slightly positive correlation with the demand level, not as pronounced as in Chengdu. On weekends, the NDL shows higher day-to-day variation than weekdays. However, NDLs are substantially more stable throughout the entire day than the NDLs in Chengdu. We speculate that the stable patterns of NDL in Pittsburgh (in terms of both time of day and day to day) are attributed to: 1) drivers' choices of departure time and routes that are relatively stable; and 2) limited roadway alternatives ({\em e.g.}, limited by tunnels and bridges).

\begin{figure}[h!]
	\centering
	\begin{subfigure}[b]{0.475\textwidth}
		\includegraphics[width=\textwidth]{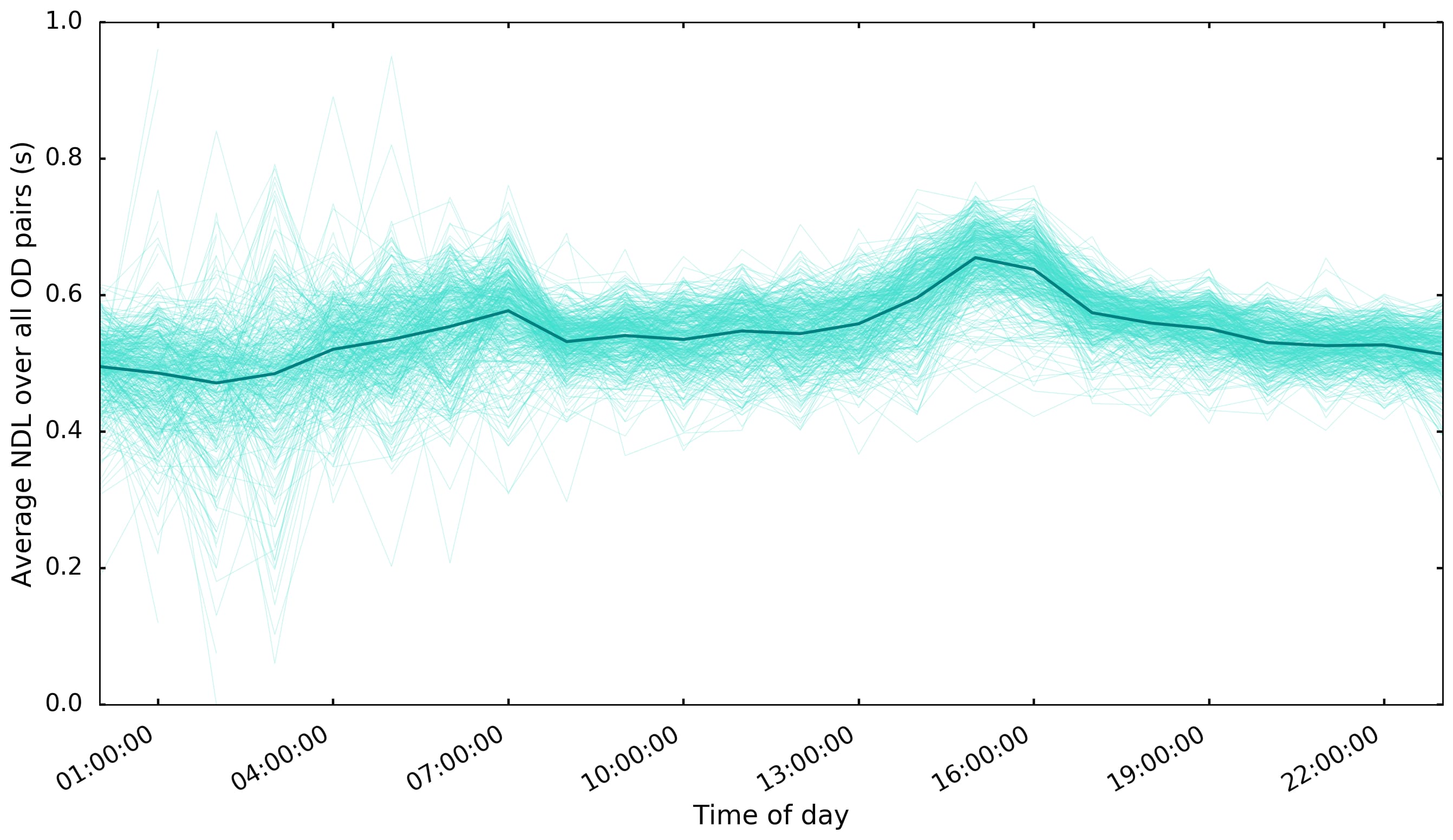}
		\caption{\footnotesize{Weekdays}}
	\end{subfigure}
	\begin{subfigure}[b]{0.475\textwidth}
		\includegraphics[width=\textwidth]{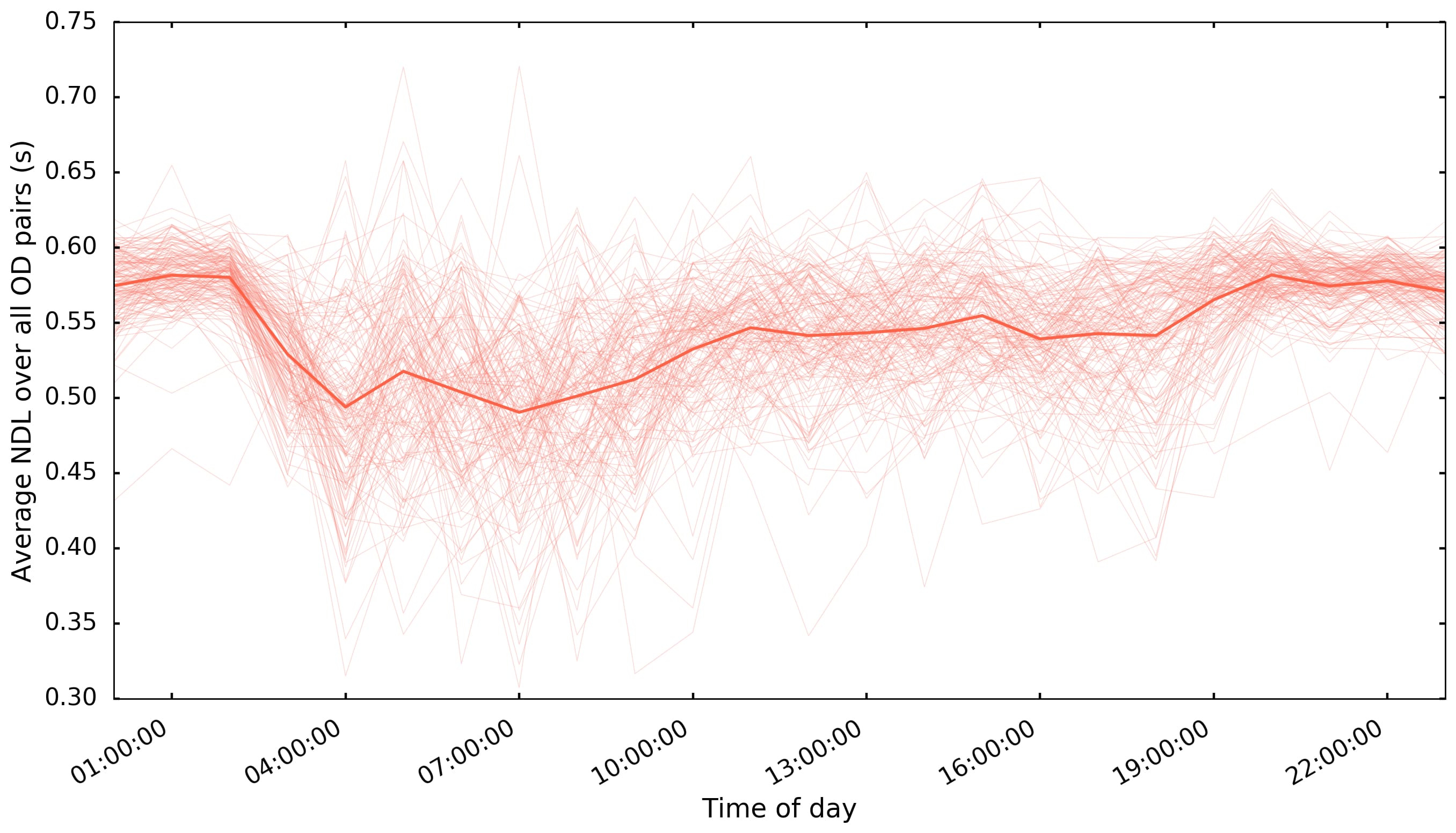}
		\caption{\footnotesize{Weekends}}
	\end{subfigure}
	\caption{\footnotesize{Average NDL by time of day, on weekdays and weekends (solid lines are the average of average NDL taken over all weekdays and weekends, respectively)}}
	\label{fig:uber_aveNDL}
\end{figure}

\subsubsection{Day of week and time of day effects on NDL measure}

Similar to section~\ref{sec:didi_time}, we plot NDL by time of day and day of week presented in Figure~\ref{fig:uber_compare}.

\begin{figure}[h!]
	\centering
	\begin{subfigure}[b]{0.7\textwidth}
		\includegraphics[width=\textwidth]{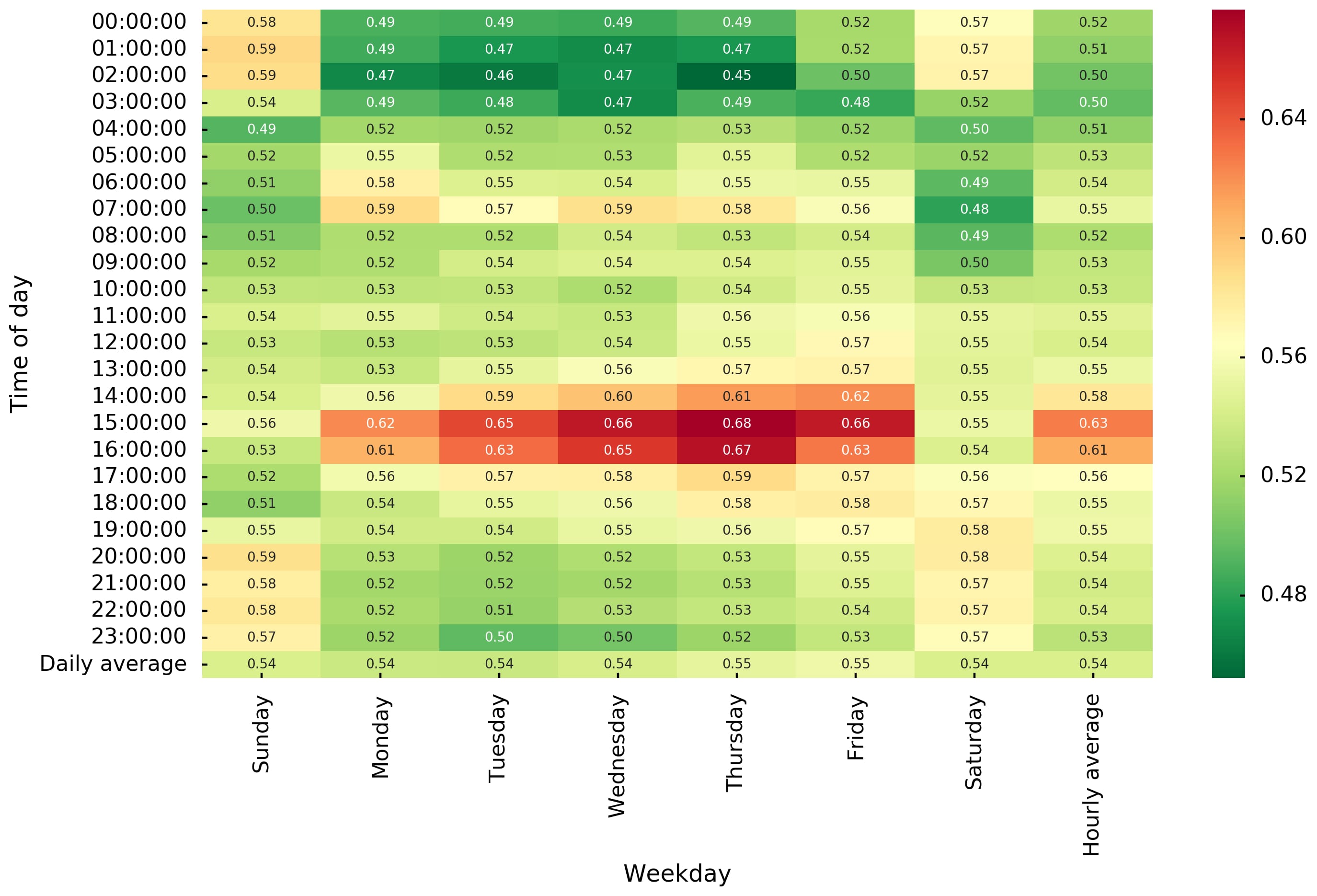}
		\caption{\footnotesize{Average NDL}}
		\label{fig:uber_compare1}
	\end{subfigure}
	\begin{subfigure}[b]{0.7\textwidth}
		\includegraphics[width=\textwidth]{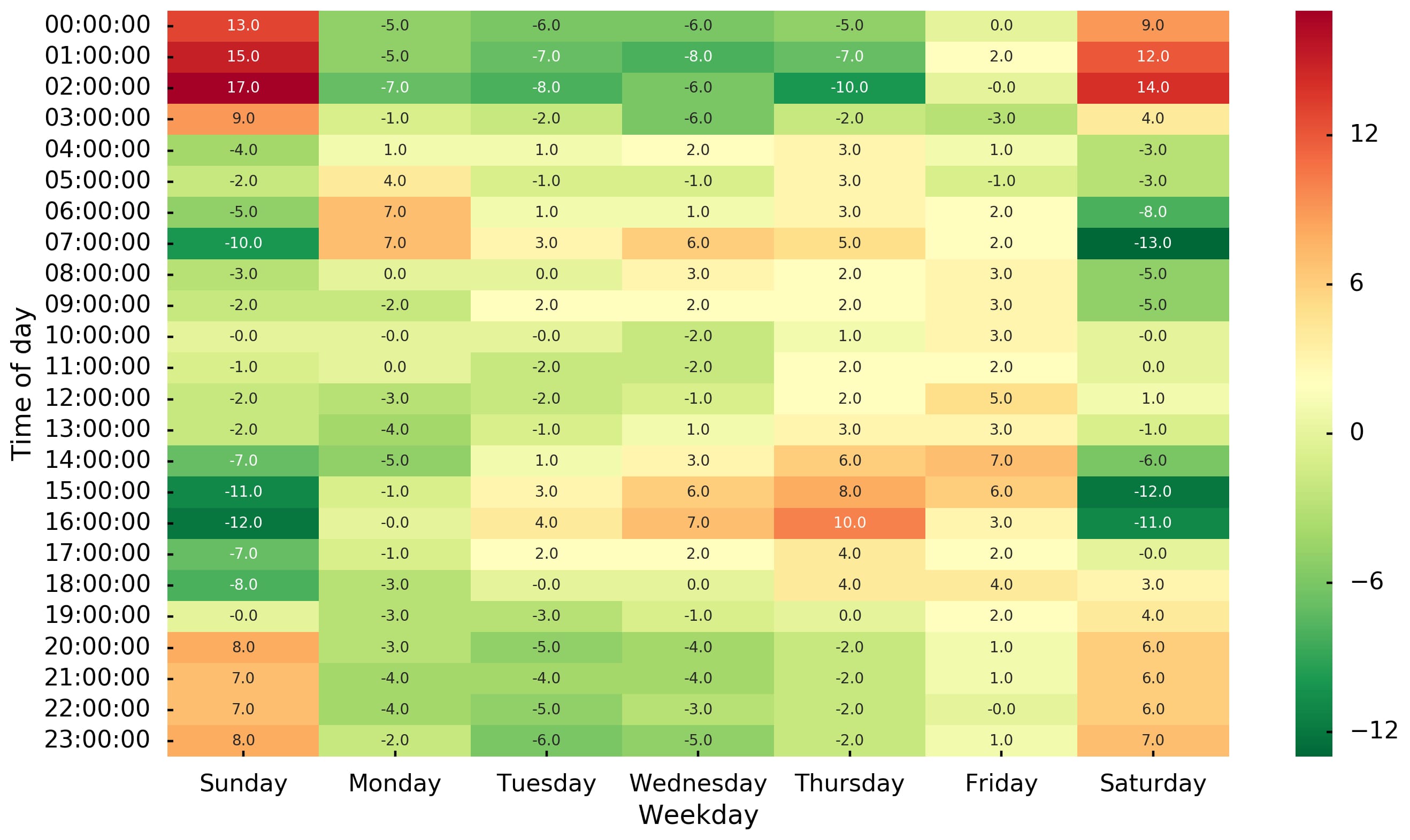}
		\caption{\footnotesize{Percentage change (\%)}}
		\label{fig:uber_compare2}
	\end{subfigure}
	\caption{\footnotesize{Average NDL and its percentage change by hour by day of week, compared to the daily average of NDL taken over
			all days of week }}
	\label{fig:uber_compare}
\end{figure}

As can be seen from Figure~\ref{fig:uber_compare1}, NDLs on weekdays are generally close from day to day, and the major spike occurs in the afternoon peak, possibly due to heavy congestion. From Figure~\ref{fig:uber_compare2}, midnight NDLs on weekends are significantly higher, possibly as a result of night activities. The NDL of early morning is the lowest over all weekends, probably because there is no congestion and trips are made straight to the destinations.

\subsection{NDL for individual OD pairs}
Now we examine the spatio-temporal NDL of each OD pair over the course of study time horizon.
\subsubsection{Time of day effects on NDL measures for a single OD pair}

Similar to Section~\ref{sec:todNDL}, we plot the time-of-day NDLs over the 18 months for randomly selected $12$ OD pairs in Figure~\ref{fig:uber_12od}. One can clearly observe the two majors spikes in the morning peak and afternoon peak.

\begin{figure}[h]
	\centering
	\includegraphics[scale = 0.15]{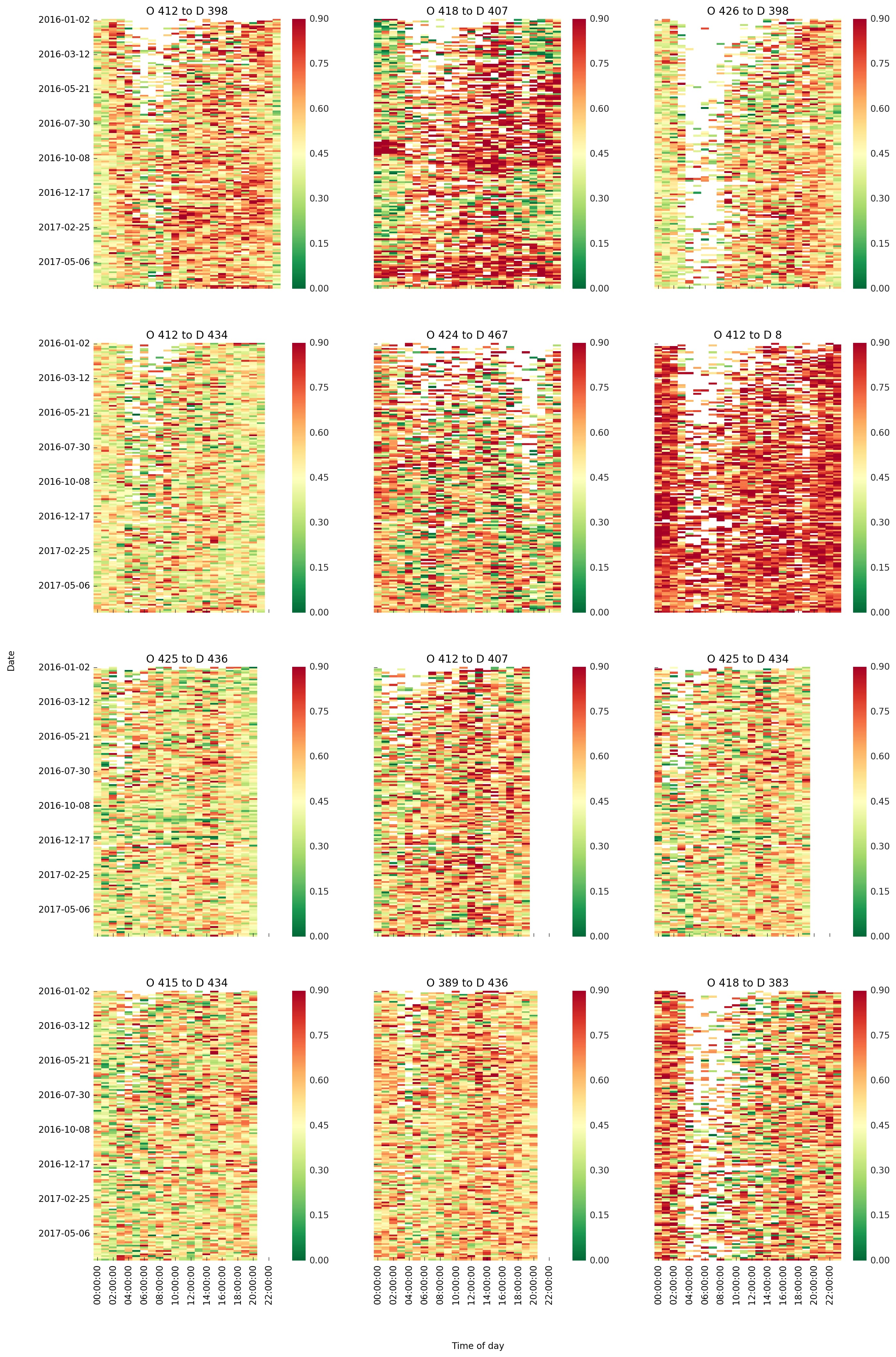}
	\caption{\footnotesize{Time-of-day NDL profile for randomly selected $12$ OD pairs}}
	\label{fig:uber_12od}
\end{figure}

There are several OD pairs with high NDL over the 18 months, such as $(418, 407)$ and $(412, 8)$. To verify the results, we further query the weekday travel times on Google Maps. The Google Maps query configurations are set such that departure time is 16:30 PM and traffic pattern is on an average Wednesday, then the results are presented in Figure~\ref{fig:example}.

\begin{figure}[h!]
	\centering
	\begin{subfigure}[b]{0.475\textwidth}
		\includegraphics[width=\textwidth]{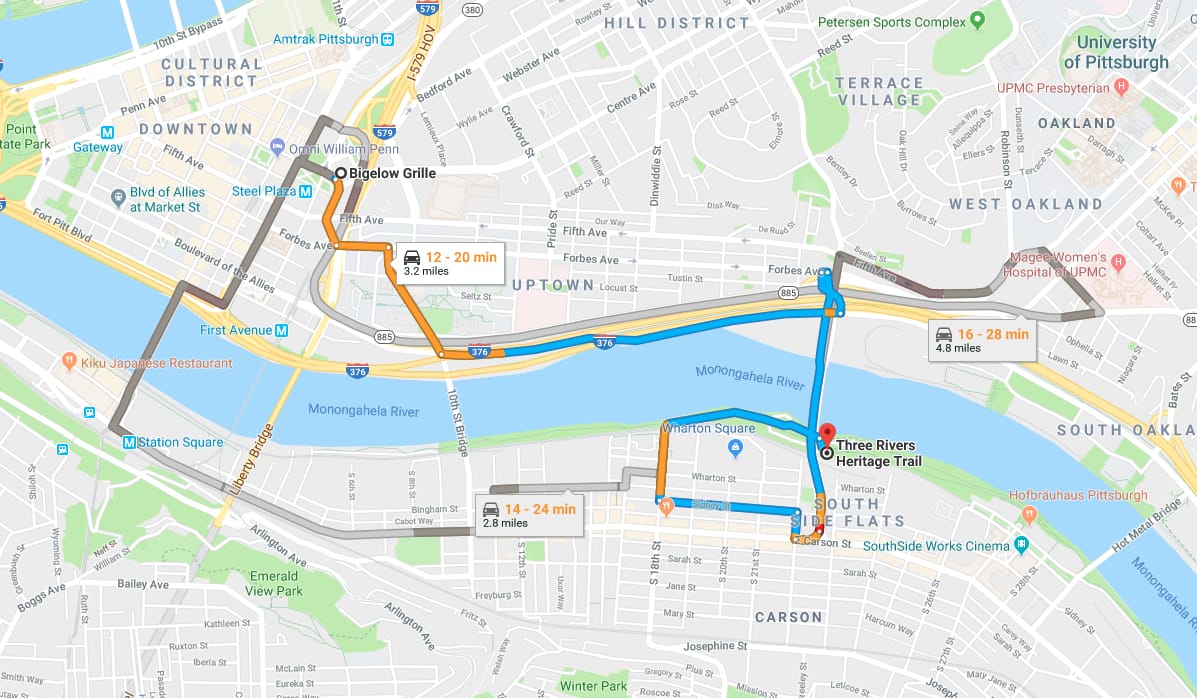}
		\caption{\footnotesize{From O $418$ to D $407$}}
		\label{fig:example1}
	\end{subfigure}
	\begin{subfigure}[b]{0.475\textwidth}
		\includegraphics[width=\textwidth]{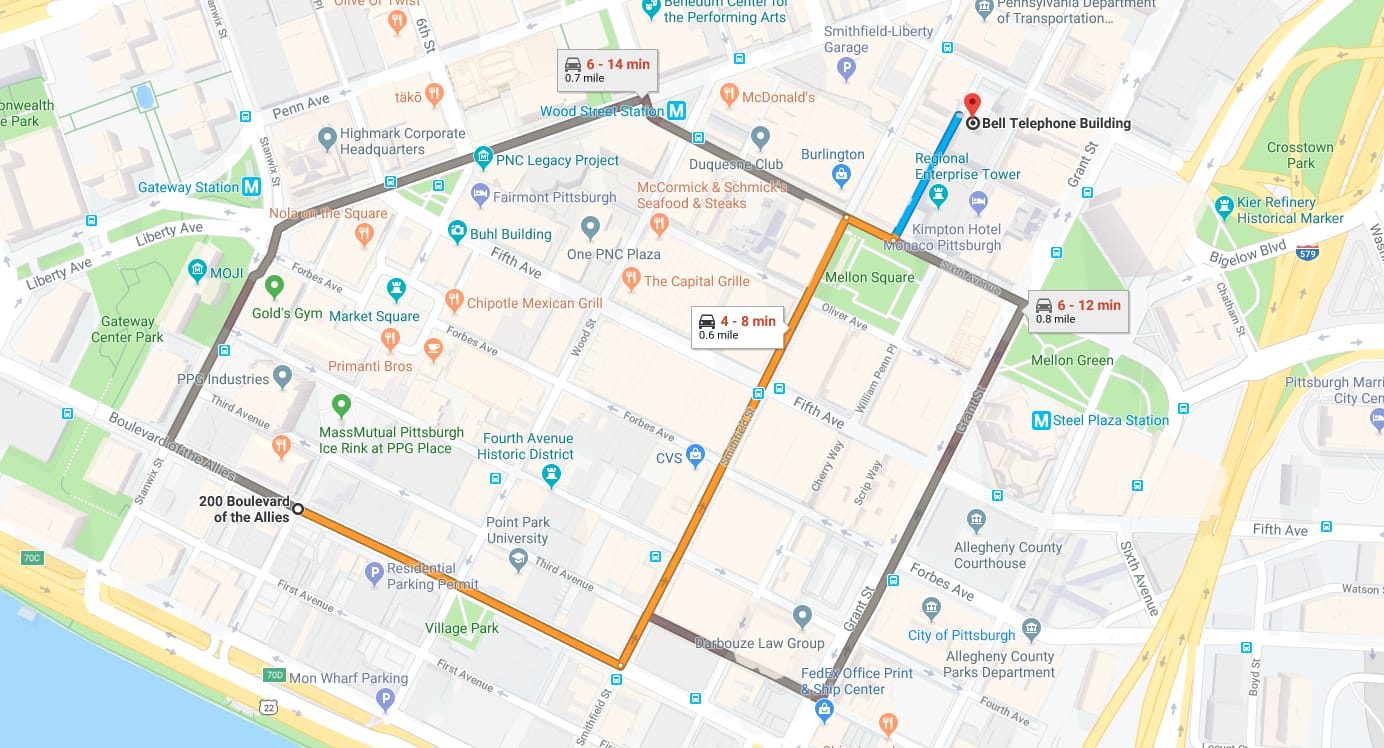}
		\caption{\footnotesize{From O $412$ to D $8$}}
		\label{fig:example12}
	\end{subfigure}
	\caption{\footnotesize{Estimated travel time and routes from Google Maps}}
	\label{fig:example}
\end{figure}

One can read from Figure~\ref{fig:example} that there are multiple routes to choose from and the variation of estimated travel time is high for both OD pairs. With high day-to-day variation, travelers can hardly decide which route is the user optimal path. It partially explains why the NDL is high for both OD pairs.

\subsubsection{Origin-based and destination-based NDL measures}

Figure~\ref{fig:uber_map} presents an overview of Pittsburgh metropolitan area and some traffic analysis zones. Again, Pittsburgh is not a city with much tourism demand, so most of its traffic is likely made by local residents.

\begin{figure}[h]
	\centering
	\includegraphics[scale = 0.2]{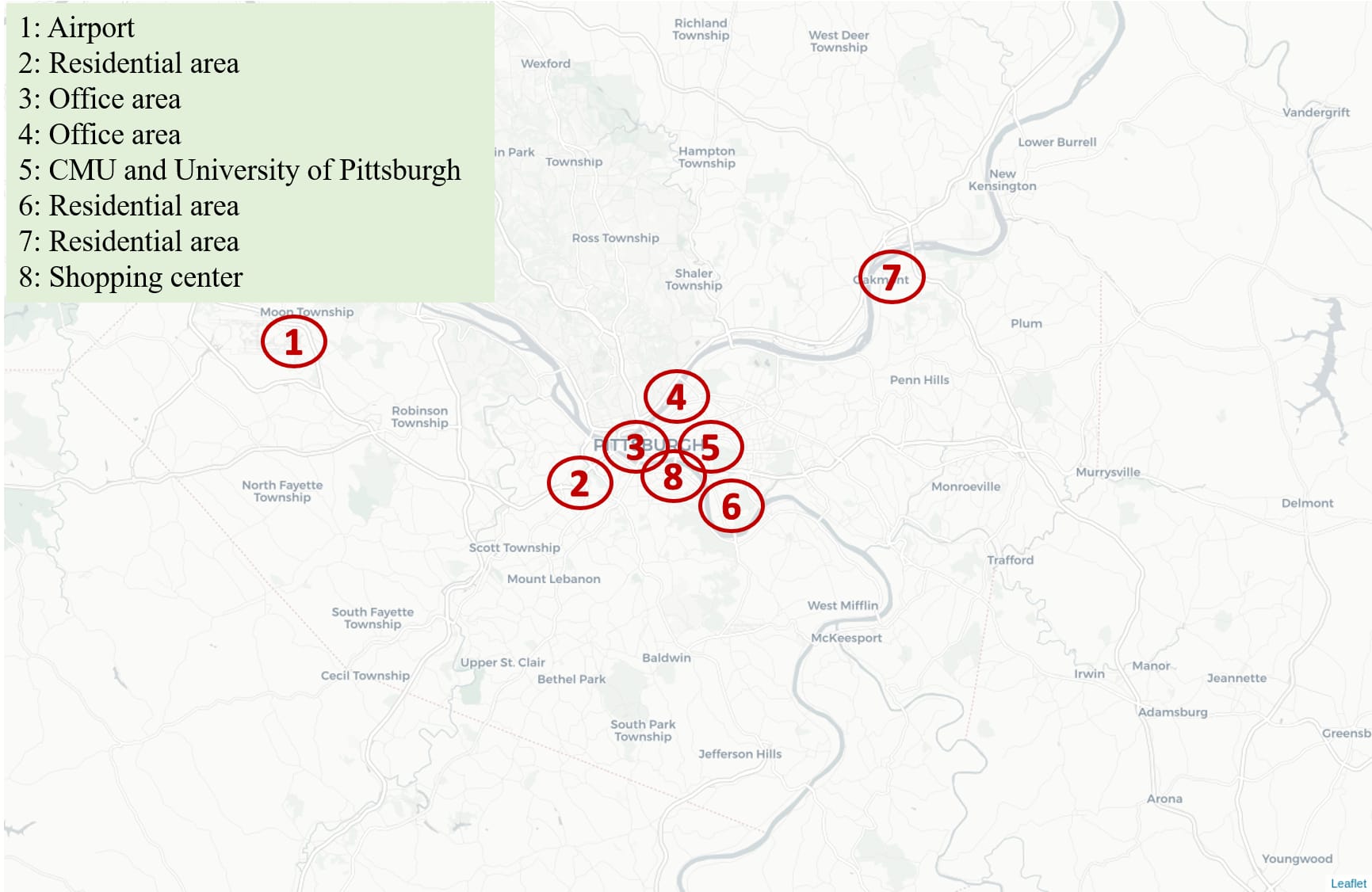}
	\caption{Overview of the Pittsburgh network and some of its main traffic analysis zones}
	\label{fig:uber_map}
\end{figure}

Similar to Section~\ref{sec:od}, we compute the origin-based and destination-based NDL by hour, averaged over all days on weekdays/weekends, and the results are presented in the supplementary materials. Four interesting origins/destinations are selected and presented in Figure~\ref{fig:uber_od_results}.

\begin{figure}[h]
	\centering
	\begin{subfigure}[b]{0.485\textwidth}
		\includegraphics[width=\textwidth]{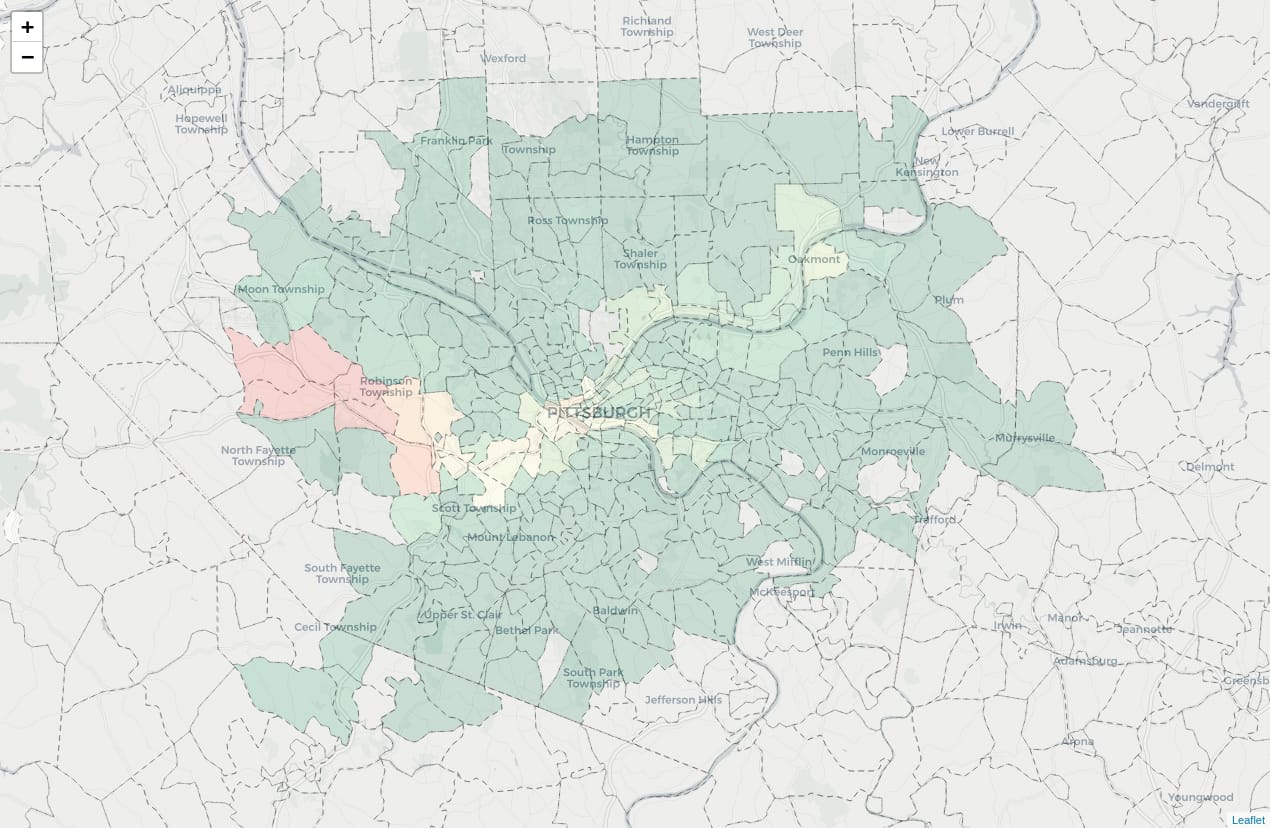}
		\caption{\footnotesize{Destination-based NDL from 4:00 AM to 8:00 AM on weekdays}}
	\end{subfigure}
	\begin{subfigure}[b]{0.485\textwidth}
		\includegraphics[width=\textwidth]{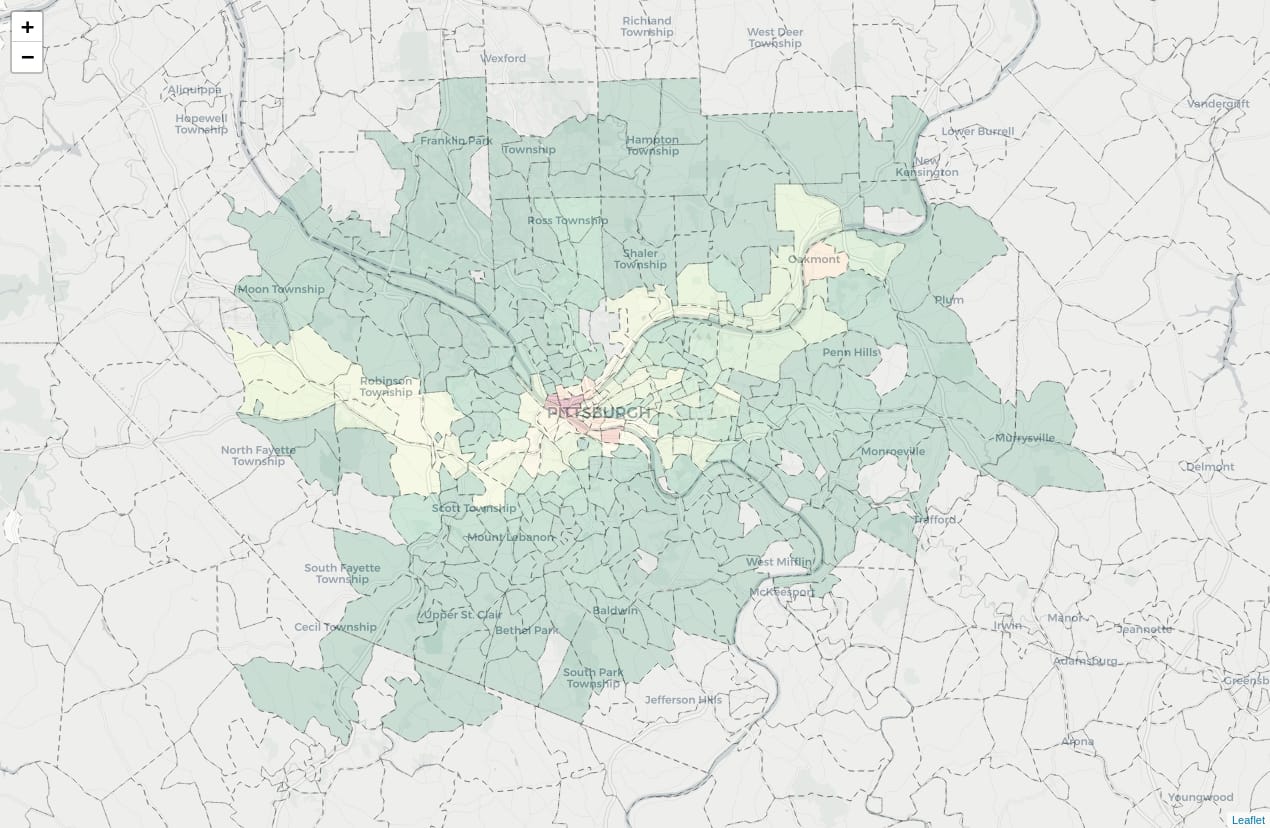}
		\caption{\footnotesize{Destination-based NDL from 8:00 AM to 12:00 PM on weekdays}}
	\end{subfigure}
	\begin{subfigure}[b]{0.485\textwidth}
		\includegraphics[width=\textwidth]{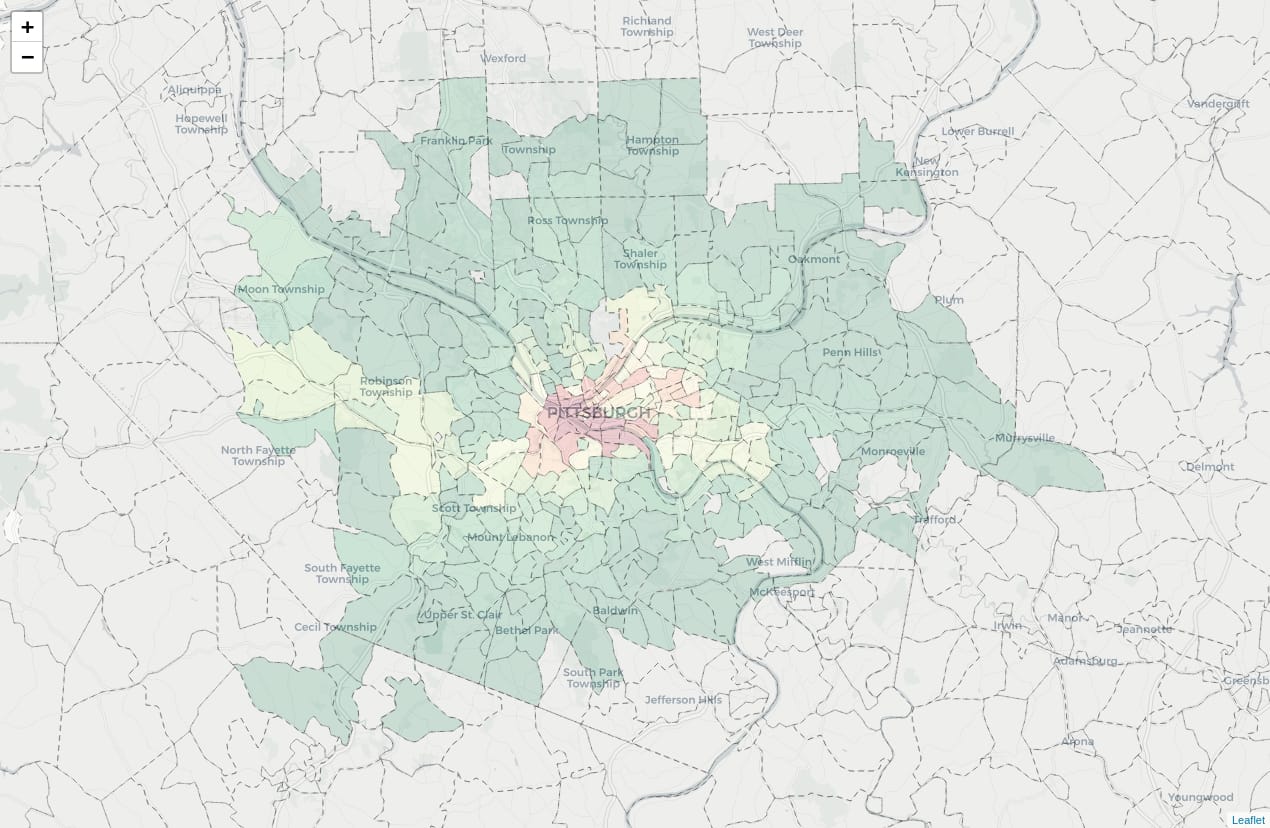}
		\caption{\footnotesize{Origin-based NDL from 0:00 AM to 4:00 AM on weekends}}
	\end{subfigure}
	\begin{subfigure}[b]{0.485\textwidth}
		\includegraphics[width=\textwidth]{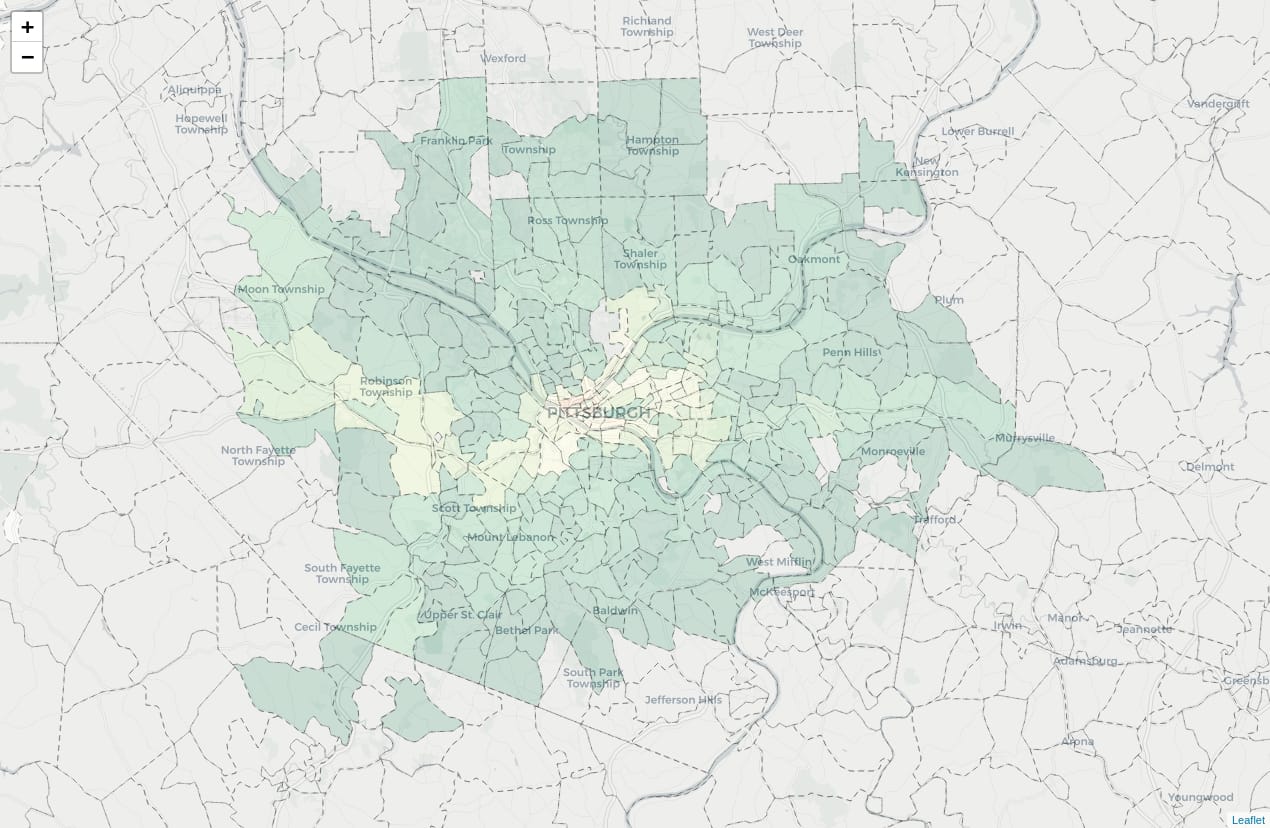}
		\caption{\footnotesize{Origin-based NDL from 8:00 PM to 12:00 PM on weekends}}
	\end{subfigure}
	\caption{\footnotesize{Selected origin and destination based NDL (Red color represents high NDL and green color represents low NDL)}}
	\label{fig:uber_od_results}
\end{figure}

Generally, the origin/destination-based NDL verifies the observations in Figure~\ref{fig:uber_aveNDL}. From 4:00 AM to 8:00 AM on weekdays, NDL is high for the airport area as destination (the zones in red). High NDL indicates there exist substantial demand during this time window and they all have multiple routes to choose from to the airport with similar travel times. From 8:00 AM to 12:00 PM on weekdays, high NDL largely occurs in the downtown area,  possibly as a result of morning commute to an area with heavy congestion. The origin-based NDL around the downtown area is also high from midnight to 4:00 AM, as a result of late night weekend activities.  In contrast, the NDL from 8:00 AM to 12:00 PM on weekends is relatively low across the entire region.

\subsection{NDL based traffic management}

Now we examine the proposed traffic management method. We have made efforts on calibrating the traffic flows, route choice models and OD demand in the Pittsburgh metropolitan area in a separate project \citep{greenfield}, and we adopt OD demand and path flow from that project, and use Uber Movement data (namely the average zone-to-zone travel time) as the expected network conditions by hour of day.

The proposed method routes the vehicles to user optimal paths. Thus, the total time can be effectively reduced if the management method can accurately detect vehicles that deviates from user optimal paths and identify their user optimal paths. The objective of this experiment is to show that NDL can potentially be used as a criterion to identify such vehicular flow and their optimal paths.  In addition, we would like to show that the proposed management method is still effective, even the estimated NDL is not up to date with an estimation delay. Since the control effects are mainly determined by whether the proposed method can find optimal user paths for a small fraction of vehicles based on NDL rather than the change in travel time if those vehicles are re-routed, we make a bold assumption that network conditions (in terms of travel time) will remain the same before and after the NDL-based control. This assumption is based on three reasons: 1) the hypothetical network conditions after control are unknown in the experiments since the field testing is too costly to do; 2) at most $5\%$ of the total demand will be re-routed, so the impact to network traffic conditions is secondary comparing to the impact of re-routed vehicles; 3) when evaluating the effectiveness of routing in terms of total travel time, the change in network conditions upon re-routing a small fraction of vehicles is likely to alleviate the total reduction (namely the effectiveness), but not the trend of reduction.  Our goal in this experiment is to show that NDL has potentials to reduce total travel time, more so than to what extent it reduces congestion. %

We study the time period 5:00 AM - 10:00 AM  from Jan. 1st, 2016 to Jun. 30th, 2017, and we assume the route choice and traffic demand are identical within all weekdays and weekends, respectively. On each day, we apply Algorithm~\ref{alg:control} with different control ratio $\alpha \in [0.01, 0.02, 0.03, 0.04, 0.05]$ and time delay of NDL estimation $\texttt{lag} \in [0, 1, 2, 3,4]$ hours at the beginning of each hour. The control effectiveness is evaluated by the reduction of total travel time of the network, and the results are presented in Table~\ref{tab:control_res}.

\begin{table}[h]
	\centering
	\caption{\footnotesize Total travel time reduction by percentage $(\%)$}
	\label{tab:control_res}
	\begin{tabular}{|c||*{5}{c|}}\hline
		\backslashbox{\texttt{lag} (hr)}{$\alpha$}
		&\makebox[3em]{0.01}&\makebox[3em]{0.02}&\makebox[3em]{0.03}
		&\makebox[3em]{0.04}&\makebox[3em]{0.05}\\\hline\hline
		0         & 13.54 & 14.15 & 14.18 & 14.19 & 14.19 \\
		1         & 7.56  & 8.05  & 8.06  & 8.07  & 8.07  \\
		2         & 5.93  & 6.22  & 6.23  & 6.23  & 6.23  \\
		3         & 4.20  & 4.27  & 4.28  & 4.28  & 4.28  \\
		4         & 3.14  & 3.39  & 3.50  & 3.55  & 3.56  \\ \hline
	\end{tabular}
\end{table}
The control with $\texttt{lag} = 0$ is the ideal situation in which the real-time NDL is immediately available. In this case, controlling $1\%$ of the vehicle will reduce the total travel time by $13.54\%$, almost as effective as re-routing 5\% of vehicles. Traffic routing with $\texttt{lag} =1$ hour is a more realistic situation in which we can obtain NDL information that is estimated one hour ago. We directly use $C_{rs}(t - \texttt{lag})$ as $C_{rs}(t)$ and the control effect is still significant. Figure~\ref{fig:control} compares total travel time before and after the NDL-based control on each day when $\texttt{lag} = 1 \text{ hour}$. It indicates that the control is effective for most of those days. With ratio $\alpha$ greater than $2\%$, the improvement in control effectiveness becomes marginal.

\begin{figure}[h]
	\centering
	\includegraphics[scale = 0.3]{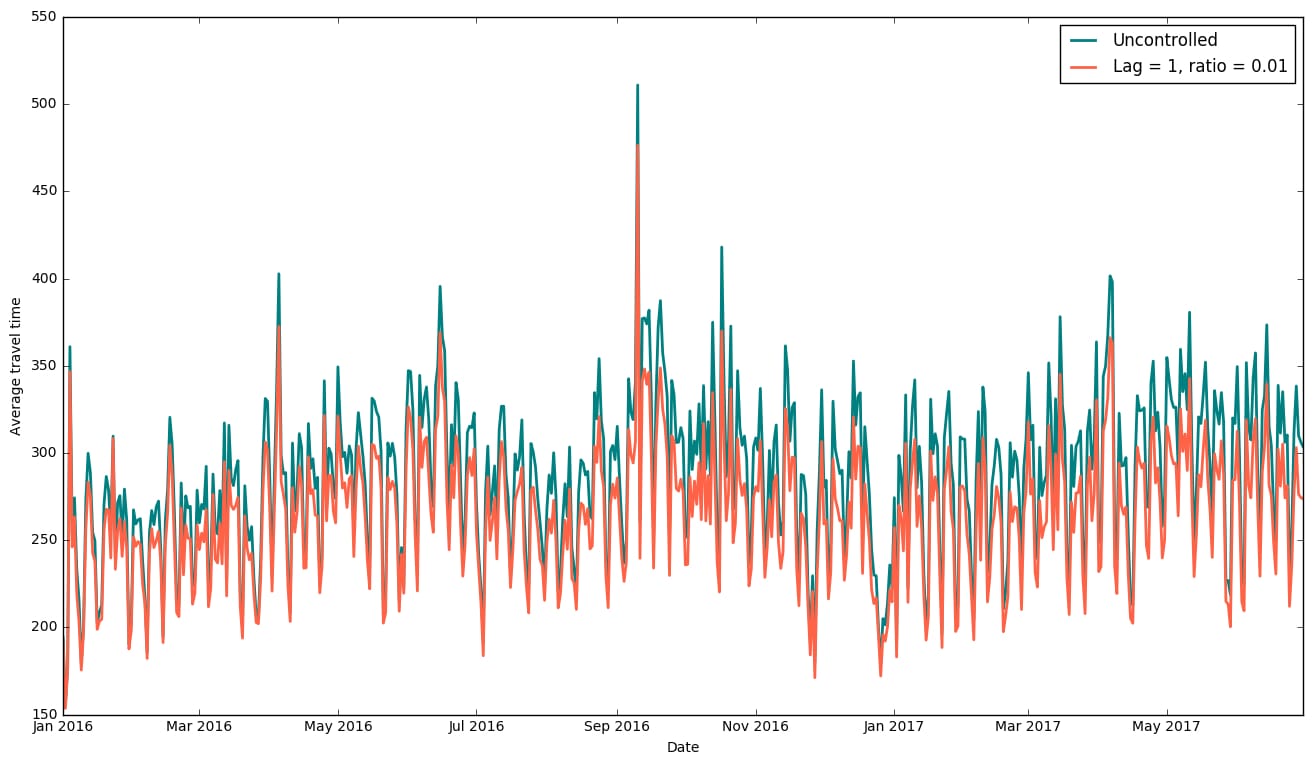}
	\caption{Comparison of before and after control in terms of total travel time (seconds)}
	\label{fig:control}
\end{figure}

\subsection{Key observations}
Combining the experiment results in Section~\ref{sec:didi} and Section~\ref{sec:uber}, we summarize the following key observations regarding the NDL measure.
\begin{itemize}
	\item The NDL does not necessarily depend on the traffic demand. Instead, it is dependent on whether the travels are recurrent or non-recurrent.
	\item High NDL is spatially and temporally sparse.
	\item Though the NDL pattern varies substantially from day to day, the average NDL pattern is stable and interpretable.
	\item Real-world networks are not under user equilibrium based on individual travel times. The use of information provision can significantly reduce traffic congestion.
\end{itemize}

\clearpage
\section{Conclusion}
\label{sec:con}

This paper first reviews the formulations for DUE and discuss how the DUE conditions may be violated in real networks. We discuss the concept of network disequilibrium level (NDL) and formulate the NDL measure based on merit function. We further discuss why NDL measure is practically implausible to estimate in real networks using conventional traffic data.

Next, this paper summarizes the characteristics of ride-sourcing vehicle (RV) data. An estimation method for the NDL is proposed with the trajectory-level (or trip-level) RV data. We propose a data sharing scheme for TNCs so that TNCs can release zone-to-zone aggregated data to the public without revealing either personally identifiable information or trip-level information that may be business sensitive. An estimation method for NDL with the zone-to-zone aggregated data is further proposed. We prove that the user optimal routing problem can be reduced to the NDL minimization problem, and an NDL-based traffic routing method is proposed. The traffic routing method basically prioritizes vehicles that deviate the most from their user optimal paths, and re-routing those vehicles may achieve effective reduction in total travel time by controlling a small fraction of those vehicles.

The NDL measure and NDL-based traffic management framework are examined in two real-world large-scale networks: the City of Chengdu with trajectory-level RV data, and the Pittsburgh metropolitan area with zone-to-zone travel time data.
We found that, for each city, NDLs are likely high when travel demand is high (thus when congestion is mild or heavy) on weekdays. Generally weekend midnight exhibits higher NDLs than weekday midnight. Many NDL patterns are different between Chengdu and Pittsburgh, which are attributed to unique characteristics of both demand and supply in each city. In Pittsburgh, weekend NDLs are generally less than weekdays except late evening and midnight, and day-to-day NDLs vary more substantially on weekend than on weekdays. NDLs of Chengdu are generally larger than Pittsburgh. In Chengdu, weekend NDLs are greater than weekdays. On the othe hand, NDL in Pittsburgh is much more stable from day to day, and from hour to hour, comparing Chengdu. It is likely attributed to 1) drivers' choices of departure time and routes that are relatively stable in Pittsburgh; and 2) limited roadway alternatives (e.g., limited by tunnels and bridges). In addition, we  observe that OD pairs with high NDL are spatially and temporally sparse for both cities. For the Pittsburgh network, we evaluate the effectiveness of NDL-based traffic routing, which shows great potential to reduce total travel time with routing a small fraction of vehicles (1\% in the experiments), even with dated NDL that is estimated in the prior hour.


In the near future, we will further explore the NDL patterns. Since NDL is closely related to the non-recurrent factors such as incidents, weather, and construction plans, we will collect non-recurrent incidents data in Pittsburgh and identify their impact to NDL. This will help us better understand the mechanism of NDL and propose more effective control strategies to traffic management.

\section*{Supplementary materials}
The browser-based dynamic network disequilibrium visualizations for Chengdu, China with DiDi Chuxing data and Pittsburgh, USA with Uber Movement data are archived at  Github\footnote{\url{https://github.com/Lemma1/Measuring-the-dynamic-disequilibrium-level}}.

\section*{Acknowledgments}
We would like to thank Uber Movement and DiDi Chuxing Gaia Initiative for providing the data. This research is funded in part by National Science Foundation CMMI-\#1751448 and Mobility 21, a national University Transportation Center on mobility funded by US Department of Transportation. The contents of this report reflect the views of the authors, who are responsible for the facts and the accuracy of the information presented herein.

\cleardoublepage
\bibliography{report}
\cleardoublepage

\end{document}